\documentclass[12pt, draftclsnofoot, onecolumn]{IEEEtran}

\usepackage[dvips]{color}
\usepackage{times}
\usepackage{epsfig}
\usepackage{graphicx}
\usepackage{amsmath}
\usepackage{amsthm}
\usepackage{amssymb}
\usepackage{amsxtra}
\usepackage{here}
\usepackage{rawfonts}
\usepackage{times}
\usepackage{url}
\usepackage{cite}
\usepackage{accents}
\usepackage{caption}
\def\tablename{Table}
\newcommand\ubar[1]{%
  \underaccent{\bar}{#1}}
\theoremstyle{definition}

\newtheorem{theorem}{\bf Theorem}

\newtheorem{proposition}{\bf Proposition}
\newtheorem{lemma}{\bf Lemma}

\newtheorem{definition}{\bf Definition}

\begin{document}
%
\title{Contract-based Incentive Mechanism for LTE over Unlicensed Channels}

\author{{Kenza Hamidouche$^{1}$}, Walid Saad$^{2}$, M\'erouane Debbah$^{1,3}$, My T. Thai$^{4}$, and Zhu Han$^{5}$\vspace*{0em}\\
\authorblockA{\small $^{1}$ CentraleSup\'elec, Universit\'e Paris-Saclay, Gif-sur-Yvette, France, Email: \protect\url{kenza.hamidouche@centralesupelec.fr}\\
$^{2}$Wireless@VT, Bradley Department of Electrical and Computer Engineering, Virginia Tech, Blacksburg, VA, USA, Email: \protect\url{walids@vt.edu} \\
$^{3}$Mathematical and Algorithmic Sciences Lab, Huawei France R\&D, France, Email: \protect\url{merouane.debbah@huawei.com}\\
$^{4}$ Department of Computer and Information Science and Engineering, University of Florida, Gainesville, FL, USA, Email: \protect\url{mythai@cise.ufl.edu}\\
$^{5}$Electrical and Computer Engineering Department, University of Houston, USA, Email: \protect\url{zhan2@uh.edu}
}\vspace*{0em}
  }

\maketitle
\vspace{-2cm}
\begin{abstract}
In this paper, a novel economic approach, based on the framework of contract theory, is proposed for providing incentives for LTE over unlicensed channels (LTE-U) in cellular-networks. In this model, a mobile network operator (MNO) designs and offers a set of contracts to the users to motivate them to accept being served over the unlicensed bands. A practical model in which the information about the quality-of-service (QoS) required by every user is not known to the MNO and other users, is considered. For this contractual model, the closed-form expression of the price charged by the MNO for every user is derived and the problem of spectrum allocation is formulated as a matching game with incomplete information. For the matching problem, a distributed algorithm is proposed to assign the users to the licensed and unlicensed spectra. Simulation results show that the proposed pricing mechanism can increase the fraction of users that achieve their QoS requirements by up to 45\% compared to classical algorithms that do not account for users requirements. Moreover, the performance of the proposed algorithm in the case of incomplete information is shown to approach the performance of the same mechanism with complete information.
\end{abstract}
\newpage
\section{Introduction}
Mobile video traffic is expected to represent 78\% of the world's mobile data traffic by 2021 \cite{cisco}. This traffic includes video conferencing, video on demand, and streaming videos which are bandwidth-intensive applications that are rapidly straining the capacity of wireless cellular systems. To support all this traffic and ensure the required application-specific quality-of-service (QoS), mobile network operators (MNOs) have started exploiting the possibility of leveraging the somewhat under-used unlicensed spectrum. This is done by enabling the base stations to offload part of their traffic over the WiFi network using the so-called LTE over unlicensed bands (LTE-U) technology \cite{zhanglte,zhang2015coexistence,liu2014small}. However, to successfully deploy LTE-U, one must ensure a fair usage of the unlicensed bands by developing new coexistence frameworks \cite{zhanglte,zhang2015coexistence,liu2014small}. Moreover, over the unlicensed bands, there is no QoS guarantee and the MNOs must incite the users to accept the best effort service based on their traffic type \cite{zhanglte}.

In this regard, several works \cite{GuanINFOCOM16,khairy2017hybrid,bairagi2018multi,bennis2013cellular,liu2014small, chen2016echo,gu2015exploiting,elsherif2015resource,kang2014mobile,chen2016impact,challita2017proactive} have addressed the problem of coexistence over the unlicensed channels so that the LTE traffic served over unlicensed spectrum does not degrade the performance of the WiFi users (WUs). 
The authors in \cite{GuanINFOCOM16} proposed a mechanism that jointly determines a dynamic channel selection, carrier aggregation and fractional spectrum access for LTE-U systems. In \cite{bennis2013cellular}, the authors proposed a cross-system learning algorithm that allows the LTE-U small base stations (SBSs) to autonomously select the unlicensed channels over which they transmit and determine their transmit power. The authors in \cite{liu2014small} formulated the load balancing between licensed and unlicensed bands as an optimization problem to maximize the overall users' rate. In the model of \cite{liu2014small}, the optimal time usage of the unlicensed channels by cellular base stations (BSs) is determined in a centralized way based on the global traffic. The work in\cite{chen2016echo} formulated the unlicensed spectrum allocation problem with uplink-downlink decoupling as a noncooperative game in which the BSs are the players that select the unlicensed channels over which they serve their users. The goal of the BSs is to optimize the uplink and downlink sum-rate while balancing the licensed and unlicensed spectra between the users. The authors in \cite{gu2015exploiting} formulated the unlicensed spectrum allocation problem as a student-project matching problem with externalities. In the model of \cite{gu2015exploiting}, a swap-matching algorithm is proposed in which the BSs propose to the users a list of unlicensed bands over which they can be served.  In \cite{elsherif2015resource}, an optimization problem is formulated to jointly allocate licensed and unlicensed bands while maximizing the users' sum-rate and ensuring fairness among these users.  

 The work in \cite{kang2014mobile} proposed a distributed traffic offloading scheme for LTE-U scenarios with a single base station. The authors in \cite{chen2016impact} proposed an economic model for the allocation of licensed band to the LTE-U BSs that can serve mobile and fixed users over unlicensed channels for free. In \cite{challita2017proactive}, the authors proposed a dynamic coexistence mechanism based on deep learning enabling multiple small base stations to select the unlicensed channels and the time they use every channel while guaranteeing fairness with WiFi users and other LTE-U operators. The work in \cite{bairagi2018multi} used multi-game theory to propose a novel coexistence framework between LTE and WiFi users. In particular, a cooperative Nash bargaining game (NBG) and a bankruptcy game were formulated to allocate unlicensed resources. The authors in \cite{xiao2018optimizing} introduced the concept of right sharing for mobile network operators allowing the operators to share and trade their spectrum access rights to the unlicensed bands. In \cite{garcia2018orla}, the authors proposed an orthogonal listen-before talk coexistence paradigm, a licensed-assisted access (LAA) scheme, that provides a substantial improvement in performance compared to classical listen-before talk mechanism, as it imposes no penalty on existing WiFi networks. The authors in \cite{khairy2017hybrid} proposed a hybrid MAC protocol to jointly optimizing the sleep period and the contention window size in LTE-U, to maximize the network performance in terms of the total network throughput and throughput fairness of LTE-U and WiFi. The work in \cite{gu2017dynamic} proposed a coexistence mechanism based on matching theory and proposed a mechanism that accounts for network dynamics and wad evaluated under two typical user mobility models. Other works \cite{mozaffari2016unmanned,challita2017network} have explored other alternatives to improve the backhaul connectivity via the deployment of unmanned aerial vehicles (UAVs). In \cite{mozaffari2016unmanned}, the authors studied the coexistence problem of device-to-device (D2D) and cellular networks in the presence of a single UAV that provides downlink transmission to support the users. The work in \cite{challita2017network} proposed a backhaul framework in which the UAVs are utilized to link the core network and SBSs when the ground backhaul is either unavailable or limited in capacity. 

Remarkably, none of these works \cite{GuanINFOCOM16,khairy2017hybrid,bairagi2018multi,bennis2013cellular,liu2014small, chen2016echo,gu2015exploiting,elsherif2015resource,kang2014mobile,chen2016impact,challita2017proactive} have addressed the economic aspect of LTE-U and all have strictly focused on the fair coexistence over the unlicensed bands. However, mobile video content that requires high QoS represents most of the traffic that is exchanged by the users over cellular network. Such QoS requirements cannot be guaranteed over the unlicensed bands due the medium access protocol that is used by the WUs and imposed to the BSs. Thus, without a suitable pricing mechanism, introducing LTE-U can create a new concern in which some users with privilege would not be incentivized to use the unlicensed bands while other users might be always offloaded to the unlicensed bands \cite{elec2015fair}. This economic concerns has been also raised in \cite{shih2018unlicensed,yu2017auction,zhang2017multi}. In \cite{yu2017auction}, the authors proposed an auction mechanism, which enables the LTE provider to negotiate with the WiFi access point owners an exclusive or shared access of the unlicensed channels. The authors in \cite{zhang2017multi} formulated the allocation problem of the unlicensed bands as a Stackelberg game in which each operator sets an interference penalty price for each user that causes interference to the WiFi access point (WAP), and the users choose the unlicensed channel and determine the optimal transmit power in the chosen channel. Hence, it is necessary to develop economic models that incite the LTE users to receive their traffic via the unlicensed channels. Moreover, the spectrum allocation mechanism must accounts for the traffic type of the users in terms of QoS requirements since no coexistence framework exists in the literature in which the requirement in terms of QoS of two separate devices is considered \cite{elec2015fair}.

The main contribution of this paper is to introduce a novel incentive mechanism for facilitating the deployment of LTE-U over cellular networks. We consider a model in which an MNO proposes agreements to a number of LTE-U users that are subscribed to its services to incite them to accept being served over the WiFi network. We consider a practical model with \emph{asymmetric and incomplete information} to account for the heterogeneity of the QoS that is required by each user and is unknown to the MNO as well as other users. The proposed approach, based on contract theory \cite{borgers2015introduction,zhang2017survey}, allows an MNO to define a contract for each user in a model with asymmetric and incomplete information, by fixing the price of serving each content and the fraction of the content that is served via the licensed and unlicensed channels. Unlike classical contract-theoretic models \cite{gao2011spectrum,duan2014cooperative,zhang2015contract,hamidouche2016breaking}, we consider a \emph{Bayesian model} with incomplete information in which the private information of a given user as well as its selected contract is not revealed to the other users. The goal of the MNO is to maximize its own reward while ensuring the feasibility of the contract it proposes to the users. For this contract model, we derive the closed-form of the price and prove its optimality and uniqueness. Moreover, we model the licensed and unlicensed spectrum allocation in an LTE-U system as a \emph{priority-based Bayesian matching game}  to determine the amount of traffic to serve over each of the frequency bands. For the formulated game, we propose a distributed algorithm to assign the licensed and unlicensed frequency bands to the users based on their priority and their QoS requirement. Then, we prove that the algorithm is guaranteed to converge to a Bayesian stable matching. Simulation results show that the proposed pricing mechanism can increase the fraction of users that achieve their requirements in terms of QoS with up to 45\% compared to a uniform pricing mechanism. Moreover, the performance of the proposed mechanism for the case of incomplete information is shown to approach the performance of the same mechanism when all the information about users' requirements is available to the MNO.

The rest of the paper is organized as follows. The system model is provided in Section \ref{model}. In Section \ref{contract}, the problem is formulated as a contract theory problem. In Section \ref{matching}, the allocation problem of licensed and unlicensed bands is formulated as a Bayesian matching game. Section \ref{sim} presents the simulation results. Conclusions are drawn in Section \ref{con}.
\section{System Model}
\label{model}
Consider a wireless network composed of a set $\mathcal{S}$ of $S$ LTE-U BSs that can serve the downlink traffic of a set $\mathcal{N}$ of $N$ users over both licensed and unlicensed bands. Since the BSs cannot guarantee any QoS when they serve users over the unlicensed bands, the MNO has to design an economic model that motivates users to accept being served over unlicensed channels depending on their QoS requirement. To define the optimal pricing policy and spectrum allocation mechanism, we use the two frameworks of contract theory and Bayesian matching games. Contract theory is a powerful mathematical tool for creating economic incentives while matching games enable to capture the interconnection between the users that share the same resources and deal with the combinatorial nature of such resource allocation problems. The pricing and spectrum allocation mechanisms are illustrated in Fig. \ref{fig:model}.

\begin{figure}
\centering
\includegraphics[scale=0.58]{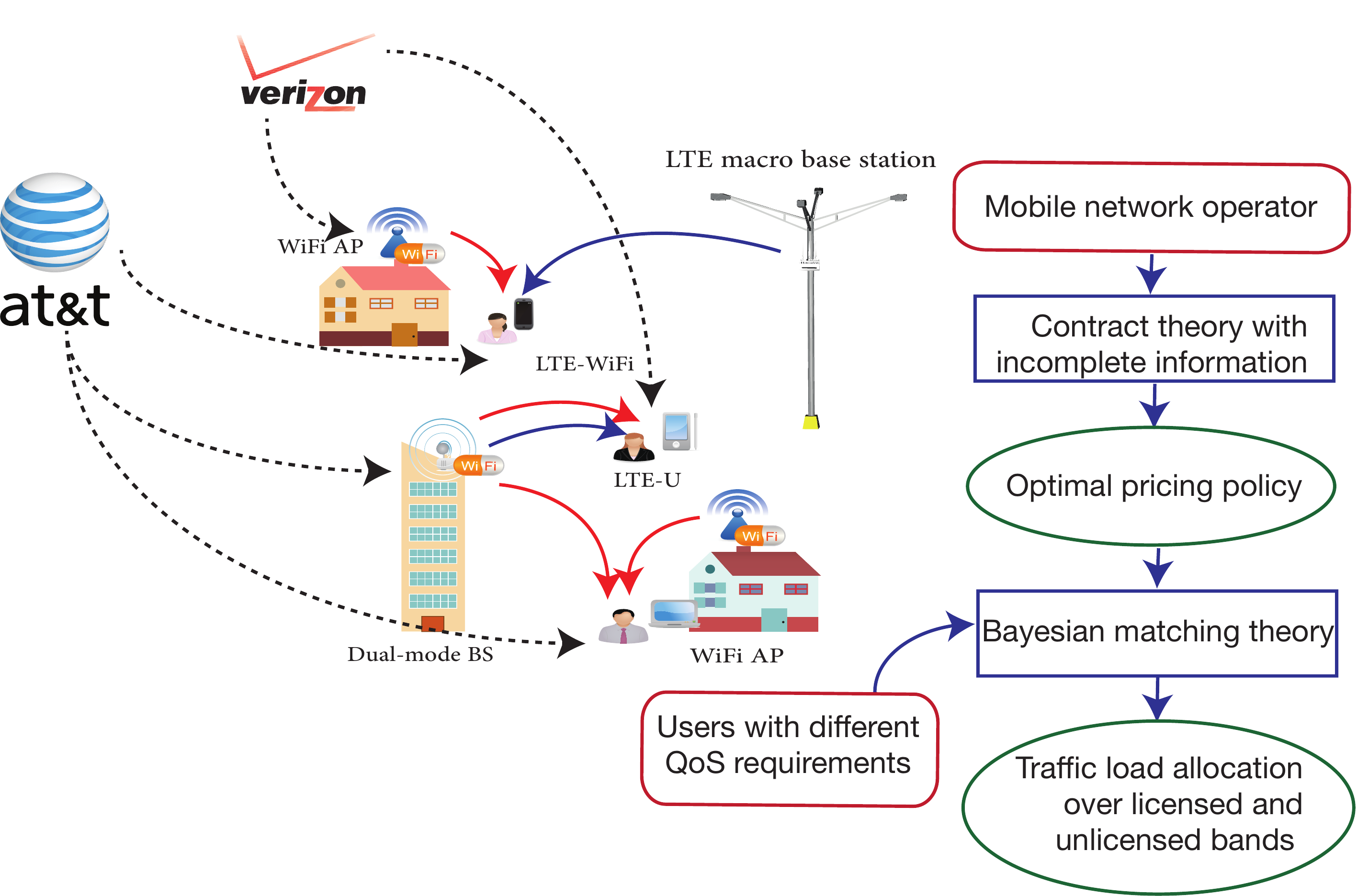}\vspace{-0.1cm}
\caption{Pricing and spectrum allocation in an LTE-U system.}\vspace{-0.3cm}
\label{fig:model}
\end{figure}

In the proposed model, each MNO specifies a performance-reward bundle \emph{contract} $(\alpha_i,\beta_i,\pi_i)$, where $\pi_i$ is the monetary reward that is charged by the MNO to user $i\in\mathcal{N}$, $\alpha_i$ is the fraction of user $i$'s traffic that is served over the licensed bands, and $\beta_i$ is the fraction of traffic that is served via the unlicensed bands. Since users have different QoS requirements depending on their traffic, we introduce a type $\theta_i\in\Theta_i=[\ubar{\theta}_i,\bar{\theta}_i]$ for each user $i$. The type represents the willingness of the user to pay for its traffic. In fact, the higher the required download rate, the more the user would be willing to pay. We sort the users' requirements in terms of QoS in an ascending order and classify them into $K$ types denoted by $\theta_1,...,\theta_K$ with $K\leq C$. Each type represents the type of the user's traffic or the willingness of a user to pay for the service offered by the MNO. This information is private to each user and is not available to the other users. The types are ordered as follows:
\begin{equation*}
\theta_1<...<\theta_k<...<\theta_K,  ~~~k\in \mathcal{K}=\{1,...,K\}.
\end{equation*}
Here, we assume that all of the requests generated by the users are launched by different users.

Instead of offering the same contract to all of the users, the MNO must offer different contracts to different users according to their types. This enables the MNO to account for the heterogeneous QoS requirements of the users and their willingness to pay for the service offered by the MNO. Moreover, the MNO must account for the interdependence between the contract designed for each type that appears in the performance achievable by every user when selecting the contract designed for its type. Such interdependence appears in the achievable rate by each user which depends on the offload policy $\{(\alpha_i,\beta_i)\}_{i\in\mathcal{K}}$ that determines the amount of traffic that is served over both licensed and unlicensed bands for all users' types. The users are free to accept or decline any type of contracts. If a user declines to receive any contract, we assume that a user signs a contract $(\alpha_i,\beta_i,\pi_i) =(0,0,0)$. Next, we define the utilities of the MNO and users in the system.

\subsection{Mobile Network Operator's Utility}
The utility of an MNO corresponds to the difference between the price charged to its users and the cost for serving them which is widely used in the literature as it captures practical considerations \cite{shah2017incentive,hajimirsadeghi2017joint,liu2017design}. The utility of the MNO when serving a user of type $\theta_i$ can be given by:
\begin{equation}
u_{\text{o},i}(\boldsymbol{\theta})=  \pi_i(\theta_i,\boldsymbol{\theta}_{-i})-c_i(\alpha_i(\boldsymbol{\theta}),\theta_i,\boldsymbol{\theta}_{-i}),
\end{equation}
where $\pi_i(\theta_i,\boldsymbol{\theta}_{-i})$ is the price charged by the MNO to a given user of type $i$, $\boldsymbol{\theta}$ is a vector that groups the types of all the users, and $\boldsymbol{\theta}_{-i}$ represents the same vector without the type of user $i$. 

We define the cost $c_i(\alpha(\boldsymbol{\theta}),\theta_i,\boldsymbol{\theta}_{-i})$ as the transmit power needed to serve a user of type $\theta_i$ with a rate $r_{si}$ from its associated BS $s\in\mathcal{S}$ which depends on the rates required by the other users. Thus, the transmit power cost can be given by:
\begin{equation}
c_i\left(\alpha_i(\boldsymbol{\theta}),\theta_i,\boldsymbol{\theta}_{-i}\right)= p_{si}(\alpha_i(\boldsymbol{\theta}))=\frac{\left(e^{\frac{r^{\text{req}}_{i}}{w}}-1\right)\left(\sigma^2_0+\sum_{j\in\mathcal{S}\setminus s}\sum_{k\in\mathcal{U}}|h_{ji}|^2\delta_{jk}p_{ji}\left(\alpha_k(\boldsymbol{\theta})\right)\right)}{|h_{si}|^2},
\end{equation}
where $r^{\text{req}}_{i}$ is the required rate of user $i$, $p_{si}(.)$ is the transmit power from BS $s$ to user $i$, $h_{si}$ is the channel gain, $w$ is the bandwidth, and $\sigma^2_0$ is the noise variance. $\delta_{jk}$ is a binary variable that is equal to $1$ when SBS $j$ is connected to user $k$ and $0$ otherwise. when We note that the transmit power $p_{ji}(\alpha_k(\boldsymbol{\theta}))$ of the other BSs depends on the type of all users as the served user is affected by the interference from all the other BSs. For notational convenience, we use $\alpha_i$, $\alpha_i(\theta_i,\boldsymbol{\theta}_{-i})$ and $\alpha_i(\boldsymbol{\theta})$ interchangeably.
\subsection{Users' Utility}
The goal of a user $i$ of type $\theta_i$ is to maximize its benefit which can be defined as the difference between the achievable rate and the price charged by the MNO:
\begin{equation}
\label{user_utility}
u_i(\alpha_i,\boldsymbol{\theta})=\theta_i v_i\left(\alpha_i(\theta_i,\boldsymbol{\theta}_{-i}),\beta_i(\theta_i,\boldsymbol{\theta}_{-i}),\theta_i\right)-\pi_i(\theta_i,\boldsymbol{\theta}_{-i}),
\end{equation}
where $v_i$ is the valuation of the user. This function can be defined as the time needed for the MNO to serve the user's traffic or equivalently, the achievable rate by the user. Note that $v_i$ is a decreasing function of the types $\boldsymbol{\theta}_{-i}$. In fact, the larger the type of other users, the higher is the traffic load on the licensed bands, and thus the QoS experienced by the user decreases. The valuation function of a user is given by:
\begin{equation}
\label{utility_sbs}
v_i(\theta_i,\boldsymbol{\theta}_{-i})=\eta (r_{si}(\alpha_i(\boldsymbol{\theta}),\beta_i(\boldsymbol{\theta}),\theta_i )-r_{i}^{\text{req}})^2,
\end{equation}
where $\eta$ is a parameter and $r_{si}(\alpha_i(\boldsymbol{\theta}),\beta_i(\boldsymbol{\theta}),\theta_i )$ represents the achievable data rate by user $i$ when served from BS $i$ over both licensed and unlicensed bands. We choose a quadratic function to model the valuation of users to capture the fact that user $i$ is only satisfied by the rate when it attains its target rate. Thus, the valuation function of a user $i$ is maximized at $r_{i}^{\text{req}}$ and decreases when the rate achieved by the user exceeds or is lower than $r_{i}^{\text{req}}$. The achievable rate by user $i$ is given by:
\begin{equation}
\label{rate}
r_{si}(\alpha_i(\boldsymbol{\theta}),\beta_i(\boldsymbol{\theta}),\theta_i )=\mathbb{E}_t\left[w_{si}\log{\Big(1+\frac{p_{si}(\alpha_i(\boldsymbol{\theta}))|h_{si}|^2}{\sigma^2+I(\alpha_i(\boldsymbol{\theta}),\boldsymbol{\theta})}\Big)}\right]+\mathbb{E}_t\left[\kappa_{si}(\beta_i(\boldsymbol{\theta}),\boldsymbol{\theta})\right],
\end{equation}
where $I=\sum_{l\in\mathcal{S}\setminus s}\sum_{k\in\mathcal{U}}{\delta_{lk}p_{li}(\alpha_k(\boldsymbol{\theta}))|h_{li}|^2}$ is the interference experienced by user $i$ from the BSs over the licensed bands, and $\kappa_{si}(\beta_i(\boldsymbol{\theta})),\boldsymbol{\theta})$ is the rate that user $i$ achieves when served by BS $s$ over the unlicensed band and is given by:
\begin{equation}
    \kappa_{si}(\beta_i(\boldsymbol{\theta}),\boldsymbol{\theta})=
    \begin{cases}
w^{\prime}_{si}\log{\Big(1+\frac{p^{\prime}_{si}(\alpha_i,\theta_i)|h^{\prime}_{si}|^2}{\sigma^{\prime 2}+I^{\prime}(\beta_i(\boldsymbol{\theta}),\boldsymbol{\theta})}\Big)}, & \text{if}\ I^{\prime}\leq I_{\text{th}}, \\
      0, & \text{otherwise},
    \end{cases}
\end{equation}
where $w^{\prime}_{si}$ is the bandwidth of the unlicensed channel, $p^{\prime}_{si}$ is the transmit power from BS $s$ to user $i$ over the unlicensed channel, and $I^{\prime}=\sum_{w\in\mathcal{S}\cup\mathcal{W}\setminus s}{p^{\prime}_{wi}|h_{wi}|^2}$ is the interference experienced by user $i$ from all the transmitting BSs and WiFi users in $\mathcal{W}$ over the same unlicensed channel. For the unlicensed channel, we define an interference threshold $I_{\text{th}}$ that is used by each BS to decide whether to transmit or not over the unlicensed channel. When a given BS senses an interference level that exceeds $I_{\text{th}}$, it does not transmit its content which results in a backoff period. Note that the interference experienced by any user depends on the types of all the other users and the amount of traffic that is served over the licensed and unlicensed bands. 

Since users do not have information about other users' types, we assume that the profile type $\boldsymbol{\theta}=[\theta_1,..., \theta_n]$ is drawn from the set $\Theta=\times\Theta_i$ according to a distribution $P(\boldsymbol{\theta})$ with density $p(\boldsymbol{\theta})$ which is common knowledge to all users. The types of users are assumed to be statistically independent such that $\phi(\boldsymbol{\theta})=\prod_{i\in\mathcal{N}}\phi_i(\theta_i)$. Moreover, considering a model with incomplete information, each user's strategy not only depends on its beliefs about $\phi(\boldsymbol{\theta})$, but also on its beliefs about the author users' beliefs about $\phi(\boldsymbol{\theta})$. Given the distribution of the users' types, each user seeks to maximize its expected utility given by:
\begin{equation} 
\label{exp_utility_users}
\bar{u}_i(\theta_i)= \int_{\Theta_{-i}} u_i(\alpha_i,\theta_i,\boldsymbol{\theta}_{-i}).
\end{equation}
The expected valuation function is defined as $\bar{v}_i(\theta_i)=\int_{\Theta_{-i}}v_{i}(\alpha(\boldsymbol{\theta}),\beta(\boldsymbol{\theta}),\theta_i)$ and the expected price is $\bar{\pi}_i(\theta_i)=\int_{\Theta_{-i}}\pi_i(\theta_i,\boldsymbol{\theta_{-i}})$. 
The optimization problem of the MNO can be defined as follows:
\begin{equation}\label{PF}
\begin{aligned}
& \underset{\{\alpha_i(.),\pi_i(.)\}_{\forall i}}{\max}
&& \int_{\Theta}\sum_{i\in\mathcal{N}}\left[\pi_i(\theta_i,\boldsymbol{\theta}_{-i})-c_i(\theta_i,\boldsymbol{\theta}_{-i})\right]dP(\theta), \\
& \text{s. t.}
&& 0\leq \alpha_i(\boldsymbol{\theta})\leq 1.
\end{aligned}
\end{equation}
In this formulation, we do not make any constraint on the participation of the users. Thus, when proposing the contracts resulting from solving (\ref{PF}), users may prefer not to select any of the contracts or select contracts that are not designed for their types. To analyze this economic incentive problem, we propose a solution based on the Nobel-prize winning framework of contract theory for designing feasible contracts \cite{zhang2017survey}. Moreover, we consider a practical case in which users have incomplete knowledge of other users' types and the contracts they select. Thus, the goal of each user is to select the most optimal contract based only on the distribution of the other users' types. 

Next, we use the framework of contract theory to determine the pricing mechanism at the MNO that incite the users to accept being served over unlicensed channels. Then, we introduce Bayesian matching theory to propose a distributed algorithm that enables the users to achieve their requirement in terms of QoS. In other words, the matching algorithm will determine the amount of traffic that should be served over both licensed and unlicensed spectrum and thus determining the values of $\alpha_i$ and $\beta_i$ for all $i\in\mathcal{U}$.

\section{Contracts Design for Unlicensed Spectrum Pricing}
\label{contract}
To ensure that the users have an incentive to collaborate with the MNO by accepting to be served via unlicensed bands, the MNO must design feasible contracts that ensure the participation of the users by at least selecting one of the contracts that the MNO proposes. Moreover, the MNO must motivate each user to select the contract designed for its own type. 

In contrast to classical contract models that are used to model resource allocation problems in wireless networks such as in \cite{gao2011spectrum,duan2014cooperative,zhang2015contract}, we consider a Bayesian model. In addition to information asymmetry that captures the unavailability of certain user types to the MNO, we assume that the type of every user as well as the contract selected by that user is not known to other users. Such information is important in our model due to the interdependence of the users' contracts that appears in the rates achievable by every user. Thus, a user has to determine it most optimal contract based only on the distribution of users types which is specific to our model. In such system, the feasibility conditions of the contracts are defined differently.

\subsection{Feasibility of a Contract}
In a Bayesian model, the MNO must design contracts for the users, that satisfy the two following constraints:
\begin{definition}
\emph{Truthfully implementable Bayesian strategies (TIBS) }: A contract $f(\boldsymbol{\theta})$ is truthfully implementable in Bayesian strategies if there exists a direct mechanism $\Gamma$ such that truth-telling is a Bayesian equilibrium. In other words, it is better for a user to select the contract defined for its type than selecting another contract, i.e.,  $\forall i\in\mathcal{N}$ and $\theta_i\in\Theta_i$,
\begin{equation}\label{IC}
\int_{\Theta_{-i}} u_i(\alpha_i,\theta_i,\boldsymbol{\theta}_{-i})\geq \int_{\Theta_{-i}}u_i(\alpha_i,\hat{\theta}_i,\boldsymbol{\theta}_{-i}),
\end{equation}
for all $\hat{\theta}_i\in\Theta_i$ and $\hat{\theta}_i\neq\theta_i$.
\end{definition}  
 This condition ensures that none of the users can achieve a better performance in the network by changing its type. The second condition ensures that the contract selected by a user should guarantee that the expected utility of the user is nonnegative given the types distribution $\boldsymbol{\Theta}_{-i}$.
\begin{definition}
\emph{Interim Individually Rational Contract (IIR)}: A contract is interim individually rational if:
\begin{equation}\label{IR}
\int_{\Theta_{-i}} u_i(\alpha_i,\theta_i,\boldsymbol{\theta}_{-i})\geq 0, ~~ \forall i\in\mathcal{N}.
\end{equation}
\end{definition}
In a Bayesian model, we can state the following lemmas from the two previous conditions.
\begin{lemma}
For any feasible contract $f(\boldsymbol{\theta})=(\alpha_i(\boldsymbol{\theta}),\beta_i(\boldsymbol{\theta}),\pi_i(\boldsymbol{\theta})$), we have $\bar{\pi}_i(\theta_i)\geq \bar{\pi}_i(\hat{\theta}_i)$ if and only if $\bar{v}_i(\theta_i)\geq\bar{v}_i(\hat{\theta}_i)$ for all $\hat{\theta}_i\in\Theta_i$.
\end{lemma}
\begin{proof}
The proof is provided in Appendix A.
\end{proof}
Lemma 1 shows that a user that pays more for its traffic should receive higher QoS or performance level compared to users that pay less, and vice versa. Next lemma shows the performance that should be achieved by the users based on the price they pay.
\begin{lemma}
For any feasible contract $f(\boldsymbol{\theta})=(\alpha_i(\boldsymbol{\theta}),\beta_i(\boldsymbol{\theta}),\pi_i(\boldsymbol{\theta}))$, if $\theta_i\geq \theta_j$, then $\bar{v}_i(\alpha_i,\theta_i)\geq\bar{v}_i(\alpha_i,\theta_j)$.
\end{lemma}
\begin{proof}
We prove this lemma using the TIBS constraint. For a user $i$ of type $\theta_i$ and user $j$ of type $\theta_j$, we have:
\begin{equation}
\theta_i \bar{v}_i(\alpha_i,\theta_i)-\bar{\pi}_i(\theta_i)\geq \theta_i \bar{v}_i(\alpha_i,\theta_j)-\bar{\pi}_j(\theta_j),
\end{equation}
and,
\begin{equation}
\theta_j \bar{v}_i(\alpha_i,\theta_j)-\bar{\pi}_j(\theta_j)\geq \theta_j \bar{v}_i(\alpha_i,\theta_i)-\bar{\pi}_i(\theta_i).
\end{equation}
By adding the two inequalities, we get:
\begin{equation}
\theta_i \bar{v}_i(\alpha_i,\theta_i)+\theta_j \bar{v}_i(\alpha_i,\theta_j)\geq \theta_i \bar{v}_i(\alpha_i,\theta_j)+\theta_j \bar{v}_i(\alpha_i,\theta_i),
\end{equation}
which is equivalent to:
\begin{equation}
(\theta_i -\theta_j)\bar{v}_i(\alpha,\theta_i)\geq (\theta_i-\theta_j) \bar{v}_i(\alpha,\theta_j).
\end{equation}

\end{proof}
Lemma 2 shows that a user of a higher type should receive a higher valuation. Note that the highest expected performance can be achieved only if its traffic is served through the licensed bands. Thus, when the application of the user requires a high data rate, the fraction of served traffic over licensed spectrum should be larger.

From Lemma 1 and Lemma 2, we conclude that for a feasible contract, all contracts satisfy:
\begin{align}
0\leq \bar{\pi}_i(\theta_{1})\leq\bar{\pi}_i(\theta_{2})\leq...\leq\bar{\pi}_i(\theta_{K}),\\ 0\leq\bar{v}_i(\alpha_i,\theta_{1})\leq\bar{v}_i(\alpha_i,\theta_{2})\leq...\leq \bar{v}_i(\alpha_i,\theta_{K}).
\end{align}

\subsection{Incentive Mechanism Design}
At the MNO, the optimization problem can be formulated as follows:
\begin{equation}
\label{problem}
\begin{aligned}
& \underset{\{\alpha_i(\boldsymbol{\theta}),\boldsymbol{\pi}_i(\boldsymbol{\theta})\}_{\forall i}}{\text{max}}
& &\int_{\Theta}\sum_{i\in\mathcal{N}}\left[\pi_i(\theta_i,\boldsymbol{\theta}_{-i})-c_i(\theta_i,\boldsymbol{\theta}_{-i})\right]dP(\theta),  \\
& \text{s.t. }
& & (\ref{IC}), \forall i\in\mathcal{N}, \\
&&& (\ref{IR}), \forall i\in\mathcal{N},\\
&&& 0\leq\alpha_i(\boldsymbol{\theta})\leq 1,\forall i\in\mathcal{N},
\end{aligned}
\end{equation}
where the first two constraints represent the TIBS and IIR conditions that ensure the feasibility of a contract. In addition to the non-convexity of the formulated optimization problem, there is a large number of constraints that increase the complexity of the problem. To solve the problem (\ref{problem}), we first state the following equivalence result.
\begin{theorem}
\label{theorem}
A contract is both TIBS and IIR if and only if $\forall i\in\mathcal{N}$,
\begin{enumerate}
\item  $\bar{v}_i(\alpha_i,\theta_i)$ is nondecreasing,
\item $\bar{u}_i(\theta_i)=\bar{u}_i(\ubar{\theta}_i)+\int_{\ubar{\theta}_i}^{\theta_i}{\bar{v}_i(\alpha,\theta_i)(\eta)d\eta},
$ for all $\theta_i$, and
\item $\bar{u}_i(\ubar{\theta}_i)\geq 0$.
\end{enumerate}
\end{theorem}
\begin{proof}
The proof is provided in Appendix B.
\end{proof}
Now, we solve (\ref{problem}) by first writing the price as a function of the amount of traffic served through the licensed and unlicensed channels and give the following result.
\begin{proposition}
\label{prop}
Let $\{(\alpha^*_i,\pi^*_i)\}_{\forall i\in\mathcal{N}}$ be a the optimal contract with fixed time allocations $\{\alpha^*_i, \forall i\}$. The optimal price for every type is given by:
\begin{flalign}\label{price}
&\bar{\pi}^*_1=\theta_1\bar{v}_1(\alpha_i^*,\theta_i),\nonumber\\
&\bar{\pi}^*_i=\theta_i\bar{v}_i(\alpha_i^*,\theta_i)+\int_{\underline{\theta}_i}^{\theta_i}{\bar{v}_i(\alpha_i^*,\eta)d\eta}, \forall i\in\mathcal{N}.
\end{flalign}
\end{proposition}
\begin{proof}
The proof is provided in Appendix C.
\end{proof}
Having defined the closed-form expression of the price charged by the MNO for every user, we have to derive the optimal values of $\alpha^*$ to define the amount of traffic that should be served over both licensed and unlicensed channels. For this, based on Proposition \ref{prop}, we write the allocation problem of the licensed and unlicensed spectra by simplifying the initial problem (\ref{problem}) in which we replace the optimal price by its expression as derived in (\ref{price}). Thus, the allocation problem of the spectrum can be defined as follows:
\begin{equation}
\label{problem_simp}
\begin{aligned}
& \underset{\{\alpha_i(\boldsymbol{\theta})\}_{\forall i}}{\text{max}}
& &\int_{\Theta}\sum_{i\in\mathcal{N}}\Big[\theta_i\bar{v}_i(\alpha_i,\theta_i)-c_i(\theta_i,\boldsymbol{\theta}_{-i})  +\int_{\underline{\theta}_i}^{\theta_i}{\bar{v}_i(\alpha_i^*,\eta)d\eta}\Big]dP(\theta),  \\
& \text{subject to}& & 0\leq\alpha_i(\boldsymbol{\theta})\leq 1,\forall i\in\mathcal{N}.
\end{aligned}
\end{equation}
The assignment problem of the traffic to the licensed and unlicensed bands is equivalent to the Knapsack problem which is NP-hard. Thus, due to the interdependence between the amount of traffic served over the licensed and unlicensed bands for the users, it is difficult to solve (\ref{problem_simp}) using classical optimization approaches. To address this problem, we propose a novel approach to assign the files requested by the users to the licensed and unlicensed bands while accounting for the QoS requirement of each file. To this end, we use the framework of matching theory with incomplete information in which the two disjoint sets of BS-channel pairs and user-subfile pairs are matched to one another by the MNO \cite{liu2014stable,wang2013interdependent}. The files are assumed to be divided into subfiles of the same size. In this model, the types of the users are not known to the BSs neither to the MNO.
\section{Frequency Band Allocation as a Bayesian Matching Game}
\label{matching}
In this section, we formulate a matching game with incomplete information which is also known as \emph{Bayesian matching game} with externalities \cite{liu2014stable,wang2013interdependent,bodine2011peer}, to determine the fraction of each requested file that should be served over both licensed frequency bands and unlicensed channels. A matching game is defined by two-sets of players that evaluate one another using preference relations. The game is considered Bayesian because the BSs do not have information regarding users' types and only have the distribution of the types. Moreover, the assignment of the channels to the users is done by the principal which is not aware of the types of the users neither of their utility functions. The matching problem is considered with externalities due to the interdependence between users preferences. In fact, the performance that is achieved by a given user depends on the global amount of traffic that is served over the same frequency band. We define traffic of the users as set of file chunks $\mathcal{F}=\{f_{11},...,f_{1K_1},...,f_{F1},...,f_{1K_F}\}$ in which each file $f_l$ of size $L_{f_l}$ bits that is requested by a user, is divided into $K_f$ file chunks of size $L_{f}^{\text{min}}$ which we define as the minimum file size that we assign to a frequency band. Then, we formulate the problem of spectrum allocation in an LTE-U system as a Bayesian one-to-many matching game in which a set of frequency bands is assigned to a set of file chunks, where each file chunk is served over a single frequency band and by a single BS. In the model we consider, the users act on behalf of the files and the BSs act on behalf of the licensed and unlicensed channels.

Thus, to model the problem as a Bayesian one-to-many matching game \cite{gu2015matching}, we consider the two sets $\mathcal{A}\in\mathcal{N}\times\mathcal{F}$ of subfile-user pairs and $\mathcal{M}\in\mathcal{C}\times\mathcal{S}$ of licensed/unlicensed frequency band and BS pairs as two sets of players. We denote the set of licensed frequency bands by $\mathcal{C}_1$ and the set of unlicensed channels by $\mathcal{C}_2$ with $\mathcal{C}=\mathcal{C}_1\cup\mathcal{C}_2$. The matching is defined as an assignment of subfiles and users pairs in $\mathcal{A}$ to BSs and channels pairs in $\mathcal{M}$ in which the BSs act on behalf of the licensed and unlicensed channels as they are the ones to decide on which frequency band they serve their traffic. Moreover, the users act on behalf of the files and determine the BS and channel to accept or not. Depending on the achievable rate and the price charged by the MNO, every user $i$ builds a preference relation $\succ_{(i,f)}$ for its file chunk $f$ over subsets of BS and channel pairs in $\mathcal{M}$ and being unmatched $\emptyset$. Similarly, every BS builds a preference relation $\succ_{(s,c)}$ \cite{bodine2011peer} for its channel $c$ over the pairs of users and subfiles in $\mathcal{A}$. In classical matching games, the players define their preferences based on a utility function that measure their performance in the network when matched to any of the players in the opposite set. In our model, the utility of any BS or user depends on the traffic served by each BS and frequency band which is a function of users' type. Therefore, in the presence of types distribution, instead of computing the exact utilities, the players define their preferences based on the expected utilities.

The number of users that a BS can serve over a given frequency band depends on the total traffic load on each frequency band as well as the channel conditions. In a matching game, the capacity of a pair $(s,c)$ is known as a quota that we denote by $q_{(s,c)}$ \cite{Roth1990}. A one-to-many matching game can be defined as follows.
\begin{definition}
A \emph{one-to-many matching} $\mu$ is a mapping from the set $\mathcal{A}\cup \mathcal{M}$ into the set of all subsets of $\mathcal{A}\cup \mathcal{M}$ such that for every $m=(s,c) \in \mathcal{M}$ and $a=(i,f) \in \mathcal{A}$ \cite{Roth1991}:
\begin{enumerate}
\item $\mu (m)$ is contained in $\mathcal{A}$ and $\mu(a)$ is an element of $\mathcal{M}$; 
\item $|\mu (m)| \leq q_{(s,c)}$ for all $(s,c)$ in $\mathcal{M}$;
\item $|\mu (a)| \leq 1$ for all $a$ in $\mathcal{A}$;
\item $m$ is in $\mu (a)$ if and only if $a$ is in $\mu(m)$,
\end{enumerate}
with $\mu (m)$ being the set of user-subfile pairs that is associated to channel $c$ and BS $s$ under matching $\mu$.
\end{definition} 
This definition states that a matching is a one-to-many relation in the sense that each channel-BS pair is matched to a set of user-subfile pairs, and each subfile is served over a single frequency band and a single BS. Before setting an assignment of subfiles to frequency bands, each player needs to specify its \emph{preferences} over subsets of the opposite set based on its goal in the network. We use the notation $\mathcal{A}_1 \succ_{(s,c)} \mathcal{A}_2$ to represent that BS $s$ prefers to serve over frequency band $c$, the user-subfile pairs in the set $\mathcal{A}_1 \subseteq \mathcal{A}$ than to serve the ones in $\mathcal{A}_2 \subseteq \mathcal{A}$. A similar notation is used for the users to set a preference list over the opposite set of BS-channel pairs. Faced with a set $\mathcal{M}$ of possible partners, a user that requests a subfile $f$ can determine a preference relation over the frequency bands and BSs pairs. We denote by $(\mu,s_1,c_1)\succ_a (\mu,s_2,c_2)$ the preference of user-subfile pair $a=(i,f)$ when user $i$ prefers being served by BS $s_1$ over frequency band $c_1$ to being served by BS $s_2$ over frequency band $c_2$ given the current matching $\mu$, with $(\mu,s_1,c_1)\neq(\mu,s_2,c_2)$. Next, we define the preferences of the users and the BSs.

\emph{1) Preferences of users:} From the users' side, each user seeks to maximize its own individual expected utility function for a given requested subfile. Therefore, from the users' point of view, we use the difference between the achievable rate and the price charged by the MNO for the user as the utility function. Thus, using the channel coefficient estimation, and the distribution of other users types, a user $i$ determines its utility for every subfile it requests $f$ given the possible BS-channel pairs using the expected utility defined in (\ref{exp_utility_users}), and thus, builds a  preference vector $\xi_a$ with $a=(i,f), \forall f\in\mathcal{F}_i$ and $\mathcal{F}_i\subseteq\mathcal{F}$ being the set of file chunks that user $i$ requests.

\emph { 2) Preferences of the base stations:} We now define a novel preference scheme at the BSs to give priority to the users based on the type they reveal for the MNO. 
A similar model is considered in \cite{semiari2014matching} in which it is shown that exploiting the information concealed in user' preferences offers considerable gains for the players. However, unlike \cite{semiari2014matching}, not only the preferences of users impact the utility of the BSs but also the types of users as discussed previously. Moreover, we account for the fact that users may manipulate the outcome of the matching process in their favor by misreporting their types \cite{wang2013interdependent}. Furthermore, unlike \cite{semiari2014matching}, the BSs can serve multiple files at a given frequency band which introduces externalities as the performance achievable by every user depends on the types of other users which indirectly determine the traffic load on each channel. 

The utility of a BS when serving a pair $a=(i,f)$ of type $\theta_a$ over frequency band $c$ given that users reveal their real types, can be given by:
\begin{equation}
\gamma_{ma}=\pi_a(\theta_a,\boldsymbol{\theta}_{-a})-c_a(\mu(\boldsymbol{\theta}),\theta_a,\boldsymbol{\theta}_{-a}),
\end{equation}
where we replace $\alpha_i$ in the cost function by the current matching $\mu$ since the fraction of traffic that is served over each of the licensed and unlicensed bands is determined by the current matching outcome. The utility of the BS represents the difference between the price that the user pays and the cost for being served at its target QoS. Knowing that users have different QoS requirements, we introduce a priority vector $\boldsymbol{\phi}\in \{0,1\}^{1\times A}$ for the user-subfile pairs in $\mathcal{A}$. This vector allows the BSs to assign priorities for the user-subfile pairs. Depending on the priority given to a user-subfile pair, a BS will promote the utility of the particular user-subfile pair. If pair $a$ has another licensed frequency band option to which it can apply according to its preference vector, the BS sets $\boldsymbol{\phi}(a)=1$, otherwise $\boldsymbol{\phi}(a)=0$. Consequently, the priority vectors are different at each BS and frequency band pair. Thus, instead of ranking users by only accounting for the cost of serving the subfile of the user and the price charged by the MNO, each BS takes into account the QoS requirements of every subfile. The utility function of each BS-channel pair $m=(s,c)$ when serving a pair $a=(i,f)$ given the types distribution of the users and the current matching $\mu$ can be defined as follows:
\begin{align}
\label{pref_m}
\xi_{ma}(I,\mu,\boldsymbol{\theta})=\int_{\Theta}\Big[\epsilon \gamma_{ma}(\theta_a,\boldsymbol{\theta}_{-a})\Big]dP(\theta)+(1-\epsilon)\Psi_{ma}(\eta_a,I(\mu(\boldsymbol{\theta}),\boldsymbol{\theta}),\kappa_{si}),
\end{align}
where the BS's utility is a function of priority coefficient $\eta_a$, resemblance factor $\epsilon\in\{0,1\}$, and users' types. Clearly, user-subfile pair $a$ will be rejected if its utility is not one of the $q_m$ highest utilities. If two users are identified by the BS to which they apply, to have the same priority and the subfiles they request are of the same type, then $\epsilon=1$. Otherwise, the BS assigns $\epsilon=0$ to the utility of those two users. Function $\Psi_{ma}$ is given by:
\begin{align}
\Psi_{ma}(\eta_a,I(\mu(\boldsymbol{\theta}),\boldsymbol{\theta}))=&\int_{\Theta}\Big[\gamma_{ma}(\theta_a,\boldsymbol{\theta}_{-a})+G_{ma}\Big] dP(\theta),\nonumber\\
=\int_{\Theta}\gamma_{ma}(\theta_a,\boldsymbol{\theta}_{-a})&+\frac{c_i(\mu(\boldsymbol{\theta}),\theta_a,\boldsymbol{\theta}_{-a})}{\theta_a\eta_a}dP(\theta).
\end{align}
The promotion function $G_{ma}$ represents the amount of promotion given to each user. It enables the BSs to prioritize serving the files with high type and less remaining options, over the licensed bands as we show later. That is, a BS decreases the cost for serving this user based on user's priority $\eta_a$ and type $\theta_a$. The higher the priority or the type of a given subfile, the more promotion it receives from the BS. Basically, let $\eta_a\in\{\eta^{\prime},\eta^{\prime\prime},\eta^{\prime\prime\prime}\}$ indicate the first, second and third priority coefficient, respectively. Parameter $\epsilon$ is used to avoid prioritizing two subfiles that have the same priority, since the promotion is a function of the users types and channel estimations knowing that a single user may have multiple requests. Clearly, the proposed priorities allow to promote subfiles that are experiencing a lower data rate compared to their required QoS, thus, allowing them to have a better BS and spectrum association. The prioritizing procedure is described as follows.

For every frequency band $c$, a BS $s$ can group the user-subfile pairs into three priority classes as follows:\\
{\bf Priority 1:} This includes user-subfile pairs that have BS $s$ and channel $c$ as their first and their only remaining preference with $c$ being a licensed frequency band. Therefore, these remaining applicants have been accepted by BS $s$ and channel $c$ in the first iteration of proposals. That is, all $a$ have $\boldsymbol{\phi}_m(a)=0$ and $\xi_a(1)=(s,c)$.\\
{\bf Priority 2:} This includes user-subfile pairs for which BS $s$ and channel $c$ is not the first preference but it is the only remaining BS in the preference list that contains a licensed band. In other words, $\boldsymbol{\phi}_m(a)=0$ and $\xi_a(1)\neq(s,c)$ and $\exists~ m_1\neq m\in\mathcal{M} \text{ with } c_1\in\mathcal{C}_1 \text{ such that } (s_1,c_1)\succ_{a} (s,c)$.\\
{\bf Priority 3:} This includes user-subset pairs that, if and when rejected by BS $s$, they still have other choices in their preference list that include a licensed frequency band.

In such practical systems in which users only know the distribution of the other users' types, the goal of the users is to select their contracts so that they reach a Bayesian stable matching (BSM) that can be defined as follows \cite{xiao2015bayesiano,xiao2015bayesian}:
\begin{definition}
The BSM represents the state in which none of the users can improve its expected utility by selecting another contract given the distribution of other users' types.
\end{definition}
To reach the BSM outcome, next, we propose a distributed matching algorithm for the assignment of the unlicensed bands to the files requested by the users.

\emph{3) Proposed LTE-U Spectrum Allocation Algorithm:} In Table I, we introduce a spectrum allocation algorithm in an LTE-U system based on the deferred acceptance algorithm \cite{gu2015matching,roth1992two}. Unlike classical one-to-many matching games with externalities in which a swap-matching is defined as an outcome of the algorithm \cite{pantisano2013matching}, the difference is in the utility functions of the players. In fact, since the distribution of users types is known to the BSs, every BS and user can determine an estimation of the interference that any user may experience over a given frequency bands. Thus, the players do not need to update their utilities and preferences based on the current matching $\mu$ at every proposals period.


\begin{table}[!t]
\label{algorithm-table}
\footnotesize
  \centering
  \caption{
    \vspace*{-0.3em}Proposed LTE-U spectrum allocation algorithm.}
    \begin{tabular}{p{14cm}}
      \hline
      \vspace{.1cm}
{\bf inputs:} $\mathcal{M}; \mathcal{A}$; $\forall i, \pi_i $; $P(\theta)$\\
\vspace{-0.1cm}\textbf{Phase 1 - Specification of the preferences}   \vspace*{.1em}\\
\hspace*{1em}- Each BS-channel pair and user-subfile pair set its preference list based on the utility functions  (\ref{pref_m}) and (\ref{exp_utility_users}), respectively.\vspace*{.1em}\\
\textbf{Phase 3 - Matching algorithm}   \vspace*{.1em}\\
\hspace*{1em}\textbf{Repeat}\vspace*{.1em}\\
\hspace*{2em}- For every subfile $f\in\mathcal{F}_i$, user $i$ sends its preference list to the next BS-band pair $m$ to which it has not applied yet. \vspace*{.3em}\\
\hspace*{2em}- Every BS $s$ updates its applicants list for every channel $c\in\mathcal{C}$ and assigns priorities to the users as defined in Section IV. 2. Then, the BS updates its preference list $\phi_m$ for every pair $m\in\{s\}\times\mathcal{C}$. \vspace*{.3em}\\
\hspace*{2em}- Every BS $s$ accepts the $q_{(s,c)}$ most preferred applicants for every channel $c\in\mathcal{C}$ and rejects the others.\vspace*{.3em}\\
\hspace*{2em}- The rejected user-subfile pairs remove the rejecting BS-channel pair from their list of preferences.\vspace*{.1em}\\
\hspace*{1em}\textbf{Until} no user-subfile pair is rejected.\vspace*{.1em}\\
{\bf Output:} Bayesian stable matching $\mu$.\\
   \hline
    \end{tabular}\label{tab:algo1}\vspace{-0.2cm}
\end{table}
\begin{proposition}
\label{theo_2}
The proposed spectrum allocation algorithm for LTE-U networks is guaranteed to converge to a Bayesian stable outcome.
\end{proposition}
\begin{proof}
The convergence of the algorithm is guaranteed due to the externalities that are canceled out in a matching game with incomplete information. In fact, given that the users estimate their utilities, they do not need to have the exact set of users and subfiles that are served over every channel. Moreover, the fact that the algorithm is based on the classical deferred acceptance algorithm guarantees the convergence to a stable outcome.
\end{proof}


%
\section{Simulation Results}
\label{sim}
 For our simulations, we consider 200 users unless stated otherwise, deployed in a $1\times1$ km area served by 20 BSs via 12 unlicensed channels and 120 backhaul resource blocks (BRBs). The users can request files of different QoS levels that we rank from 1 to 6 for the corresponding values $[200,250,350,450,550,650] \text{ Mbps}$. The network parameters are provided in Table II.
\begin{center}
\scriptsize
   \begin{table}[t!] 
\caption{ Simulation parameters.}\label{tab:table}
 \begin{tabular}{|c|c||c|c|} 
 \hline
Parameters & Values & Parameters & Values \\ 
 \hline
  Total number of BSs & $20$  &  Required rate per user& $\{0,0.2,0.4,0.5,0.6,0.7 \}$Mbps  \\
 \hline
  Total number of WAPs& 10 & WAP communication range & $90$ m \\
 \hline
Total number of unlicensed channels & $12$ &  BS communication range& $ 200$ m  \\ 
 \hline
Total number of licensed resource blocks & 120 &  Path loss exponent & 3  \\
 \hline
  Noise power& -174 dBm/Hz&  Total number of files & 100 \\ 
 \hline
Total bandwidth& 1 GHz & Size of a file   & 50 Mbits \\ 
 \hline

\end{tabular}
\end{table}
\end{center}
\vspace{-1cm}

Fig. \ref{res1} shows the mean achievable rate by the users in the scenario with complete information in which case the BSs know the rate requirement of the users and the Bayesian framework we considered. The performance of the users in the proposed mechanism approaches their performance in the case of complete information. The data rate of the users decreases in the network due to the interference that is larger when the number of users increases in the network. Moreover, the allocated resource blocks per user decreases as the number of users increases.

\begin{figure}
\centering
\includegraphics[scale=0.65]{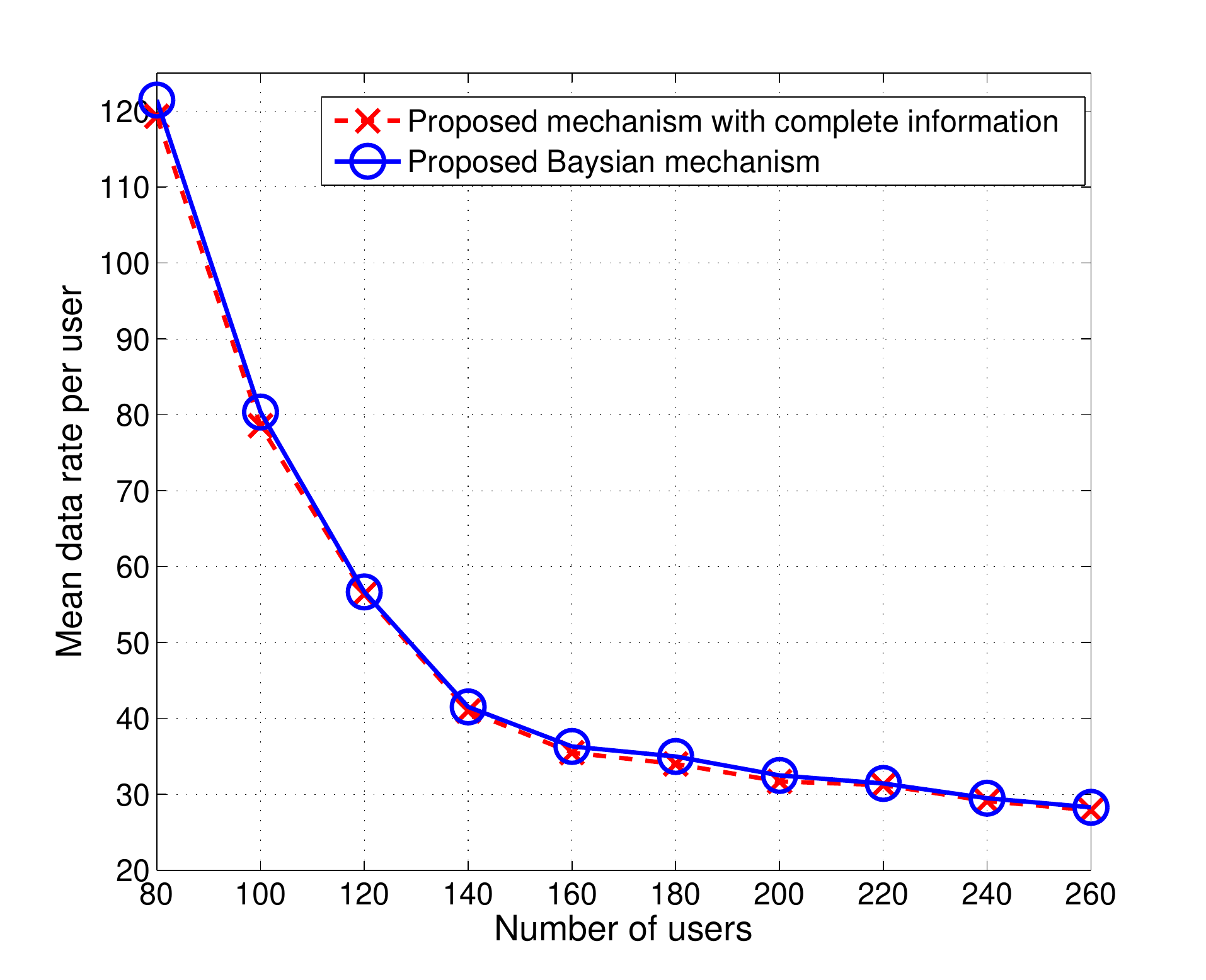}\vspace{-0.1cm}
\caption{Numerical results for the mean data rate per user as a function of the number of users.}\vspace{-0.3cm}
\label{res1}
\end{figure}
In Fig. \ref{res2}, we show the impact of the evolution of users' utility when they select different contracts. A given user achieves its highest utility when selecting the contracts designed for their own types. Indeed, when a user selects a contract of a type higher than its own type, the user cannot afford the resource blocks that are allocated to serve its content. However, when a user selects a contract of lower type, the resource blocks that are allocated to serve its content are not sufficient to meet its QoS requirements.
\begin{figure}
\centering
\includegraphics[scale=0.65]{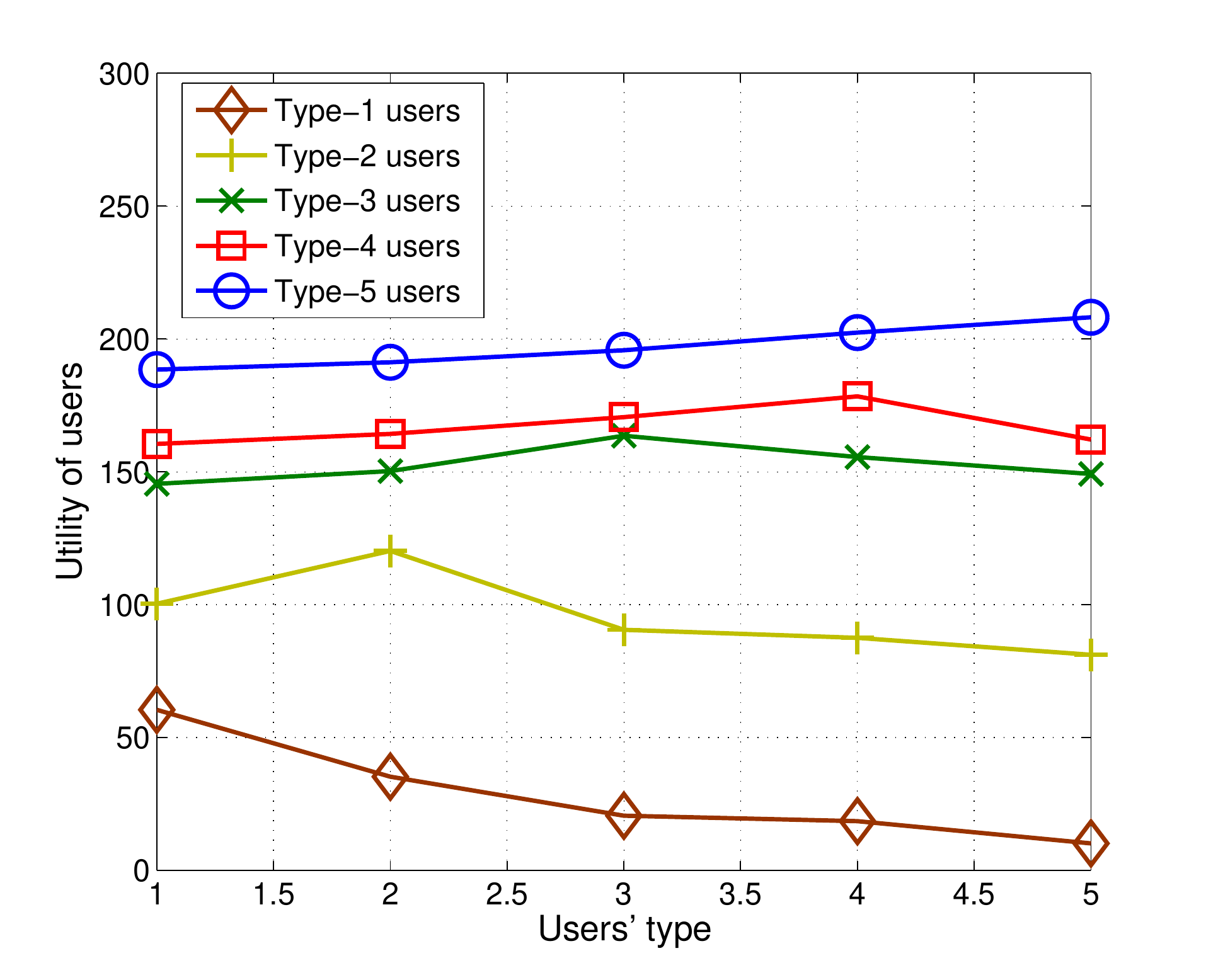}\vspace{-0.1cm}
\caption{Users utility resulting from the proposed approach as a function of contracts' type.}\vspace{-0.3cm}
\label{res2}
\end{figure}

In Fig. \ref{res5}, we compare the amount of traffic that is served over unlicensed bands for the type 1 users and type 3 users. Depending on the case we consider, we add type 3 or type 1 users to the network. The results show that the amount of traffic offloaded to the unlicensed bands increases for both type 1 and type 3 users as the number of users increases in the system. However, more type 1 traffic is served over unlicensed bands as the requirement in terms of data rate is less critical compared to type 3 traffic and the amount of offloaded traffic remains the same when the number of users exceeds 600 due to the saturation of the unlicensed bands. 

\begin{figure}
\centering
\includegraphics[scale=0.65]{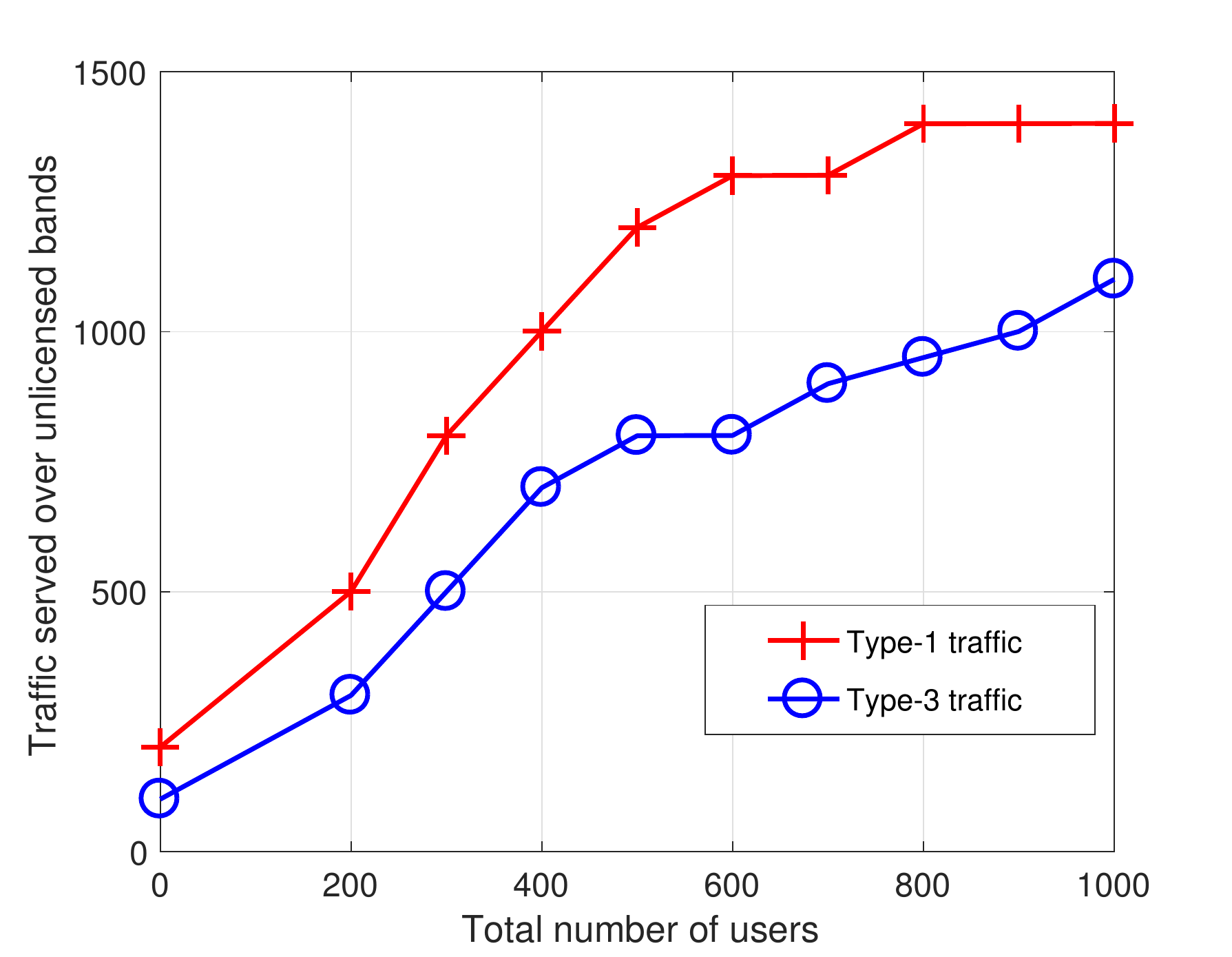}\vspace{-0.1cm}
\caption{Amount of traffic offloaded to the unlicensed bands as a function of the number of users.}\vspace{-0.3cm}
\label{res5}
\end{figure}

Fig. \ref{res3} shows the fraction of users that achieve their QoS requirement as a function of the number of users in the networks. In Fig. \ref{res3}, we compare the proposed mechanism with an allocation framework in which a random number of users accept to be served over the unlicensed channels. The fraction of users that achieve their QoS decreases by increasing the number of users in the networks because the licensed and unlicensed channels have to be shared among a large number of users which decreases their data rate. The proposed mechanism outperforms by up to 45\% the unaware QoS resource allocation algorithm. This is due to the random allocation of the licensed and unlicensed bands that does not account for the requirements of the users. Thus, the random allocation can result in users with high QoS requirements served via unlicensed channels.
\begin{figure}
\centering
\includegraphics[scale=0.6]{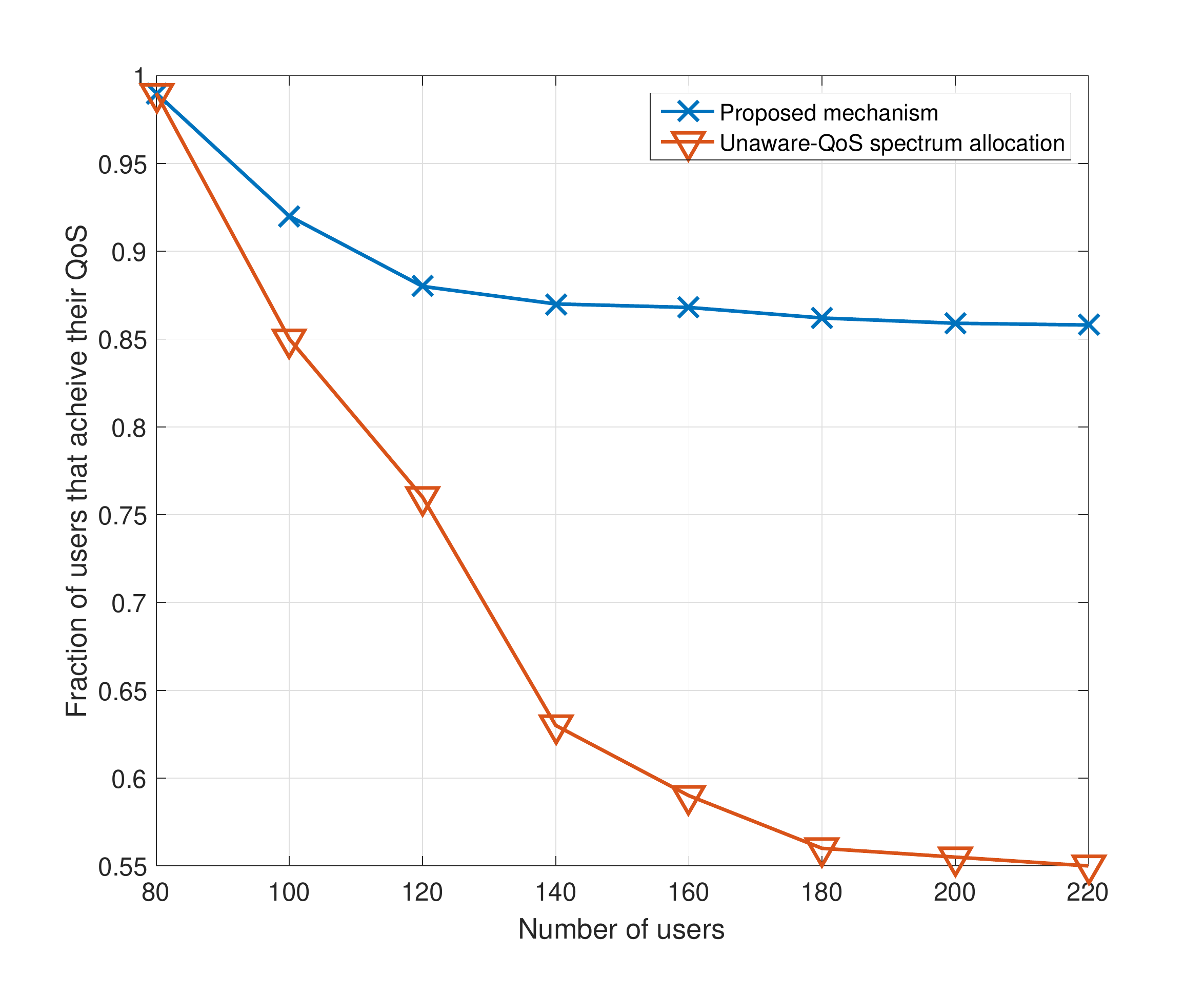}\vspace{-0.1cm}
\caption{Fraction of users that achieve their required QoS as a function of the number of users.}\vspace{-0.3cm}
\label{res3}
\end{figure}

In Fig. \ref{res4}, we compare the proposed contract-based pricing scheme with a uniform pricing mechanism in which the same price is fixed for serving the files. The results in Fig. \ref{res4} show the evolution of the users' utility as a function of the number of users in the system. The experienced utility per the users decreases when increasing the number of users. This is due to the interference from the BSs that impact the achievable data rate per user. When the network has more than 600 users, the difference between the two schemes becomes substantial as the average utility per user becomes twice higher than the utility they achieve under the uniform pricing scheme. This is due to the prices fixed by the MNO that are adapted to the requirement of the users which incite them to select the contract that offers the best possible data rate by accounting for other users requirements and traffic load. 

\begin{figure}
\centering
\includegraphics[scale=0.7]{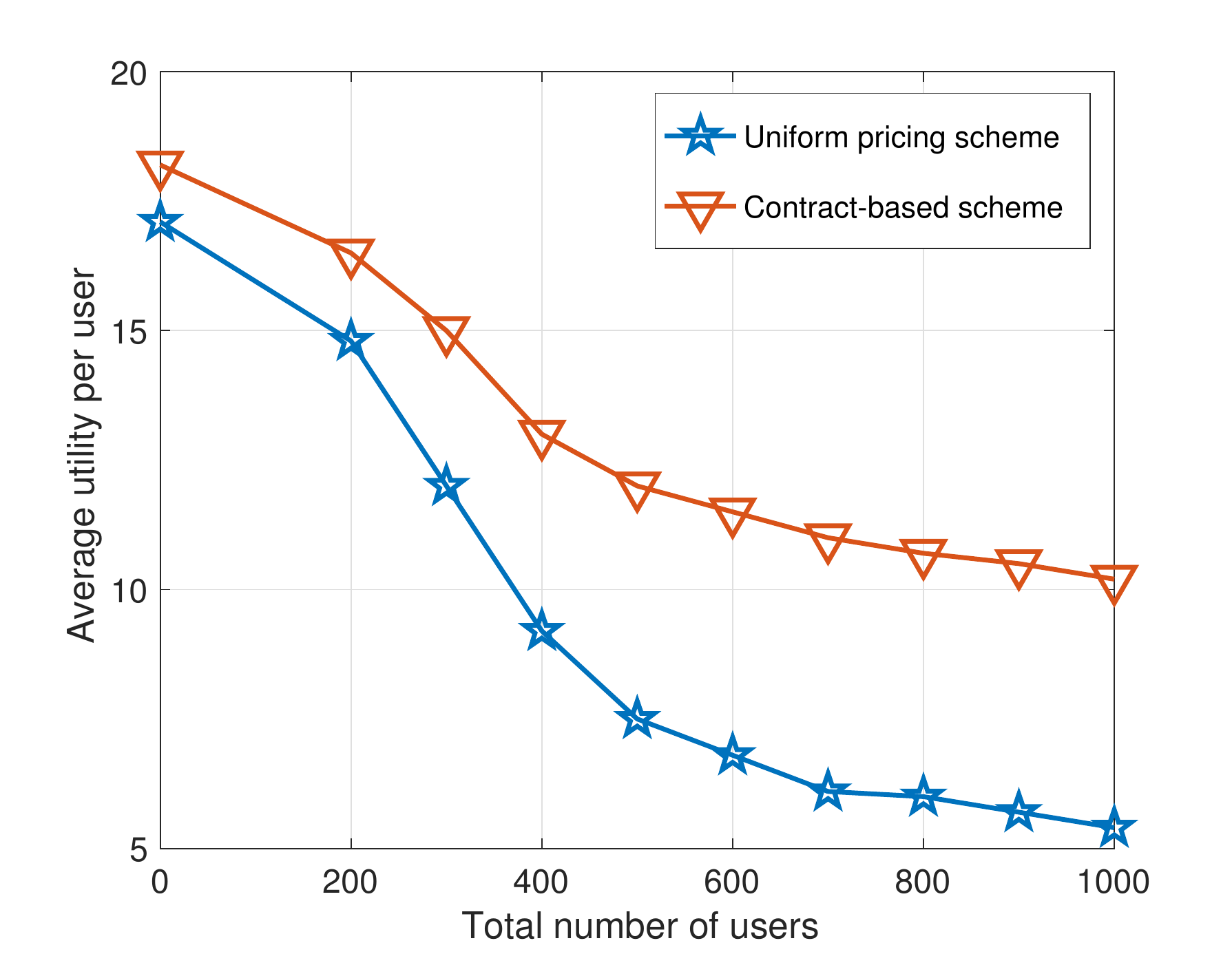}\vspace{-0.1cm}
\caption{Users' utility resulting from the proposed approach and a uniform pricing scheme.}\vspace{-0.3cm}
\label{res4}
\end{figure}

\begin{figure}
\centering
\includegraphics[scale=0.7]{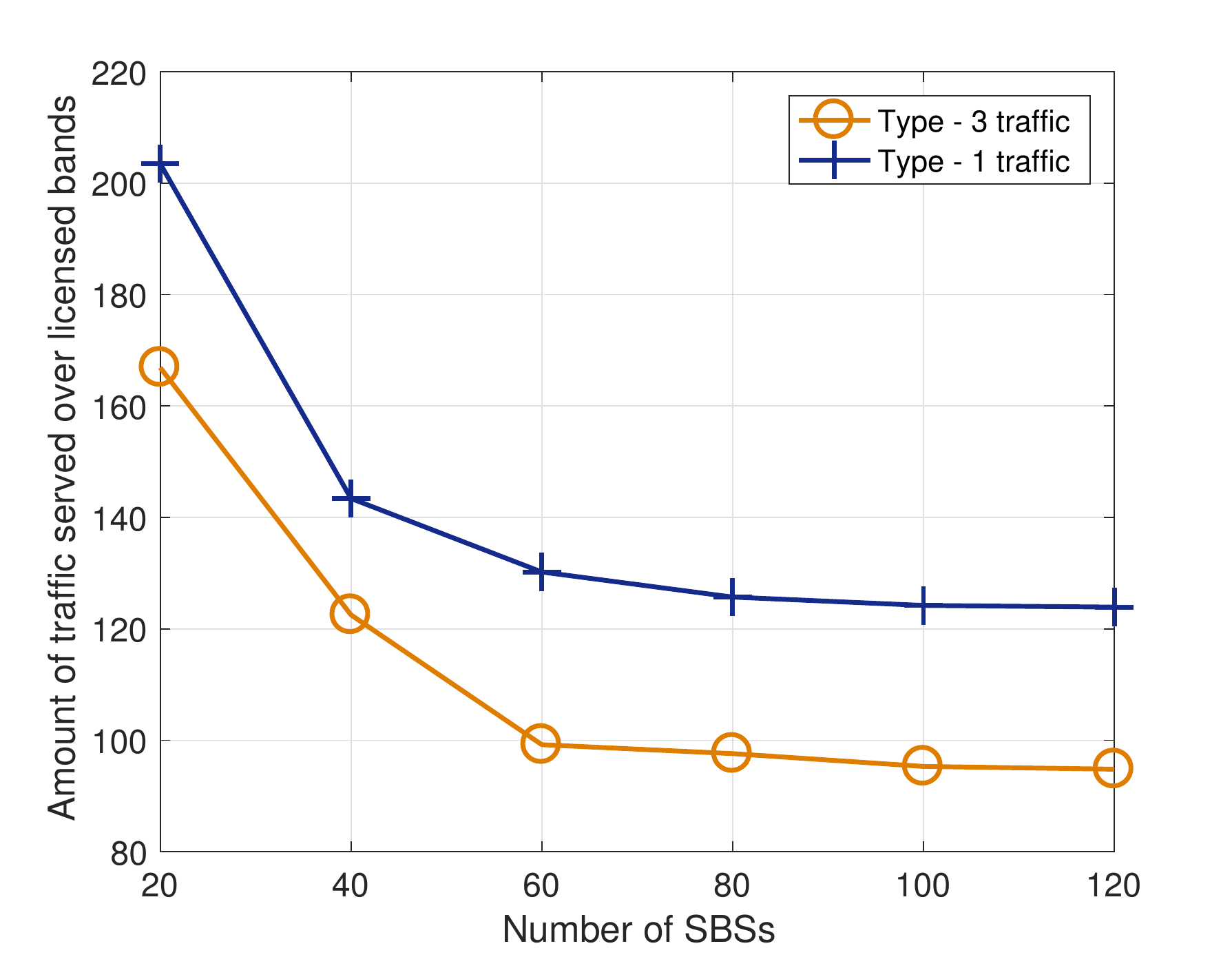}\vspace{-0.1cm}
\caption{Amount of traffic served over licensed bands as a function of the number of BSs.}\vspace{-0.3cm}
\label{res6}
\end{figure}
In Fig. \ref{res6}, we compare the amount of traffic that is served over licensed bands for type 1 users and type 3 users as a function of the number of BSs in the network. The results show that the amount of traffic served via the licensed bands decreases initially for both type 1 and type 3 users as the number of BSs becomes higher in the system. However, more type 1 traffic is served over licensed bands as the data rate requirement of type 1 users is less critical compared to type 3 traffic. Moreover, Fig. \ref{res6} shows that the amount of traffic served over the unlicensed bands remains relatively the same for both type 1 and type 3 users when the number of BSs exceeds 60 BSs. This is due to the limited licensed spectrum resources that are shared among all of the BSs.

\section{Conclusion}
\label{con}
In this paper, we have proposed a novel LTE-U incentive mechanism that motivates users to accept being served over unlicensed channel. To this end, we have used contract theory enabling every MNO to set a pricing mechanism for LTE-U users and then used matching theory to determine the fraction of LTE traffic that should be offloaded to the unlicensed bands while accounting for the QoS requirements of every content. The main advantage of using contract theory is the ability to model information asymmetry and incompleteness given that it is difficult for MNOs to gather information about users requirements. The outcome of the contract mechanism is a contract defined by the price charged by the MNO for serving the content as well as the amount of resource that is required to achieve the required QoS by the users. Moreover, we have proposed a distributed matching algorithm to allocate both licensed and unlicensed frequency bands to the users. Simulations results have shown that the traffic can be better allocated using the proposed mechanism and increase the fraction of users that are served at their requested QoS.

\begin{appendices}
\section{Proof of Lemma 1}
We divide the proof into two parts. First, we prove that if $\bar{\pi}_i(\theta_i)\geq \bar{\pi}_i(\hat{\theta}_i)$, then $\bar{v}_i(\theta_i)\geq \bar{v}_i(\hat{\theta}_i)$. Due to TIBS constraint in (4), we have:
\begin{equation}
\theta_i \bar{v}_i(\alpha,\theta_i)-\bar{\pi}_i(\theta_i)\geq \theta_i \bar{v}_i(\alpha,\hat{\theta}_i)-\bar{\pi}_i(\hat{\theta}_i),
\end{equation}
i.e.,
\begin{equation}
\theta_i\Big[ \bar{v}_i(\alpha,\theta_i)-\bar{v}_i(\alpha,\hat{\theta}_i)\Big]\geq -\bar{\pi}_i(\theta_i) -\bar{\pi}_i(\hat{\theta}_i).
\end{equation}
Since $\bar{\pi}_i(\theta_i)\geq\bar{\pi}_i$, we conclude:
\begin{equation}
\theta_i\Big[ \bar{v}_i(\alpha,\theta_i)-\bar{v}_i(\alpha,\hat{\theta}_i)\Big]\geq -\bar{\pi}_i(\theta_i) -\bar{\pi}_i(\hat{\theta}_i)\geq 0.
\end{equation}
and thus, $\bar{v}_i(\alpha,\theta_i)\geq \bar{v}_i(\alpha,\hat{\theta}_i)$. Next, we prove that if $\bar{v}_i(\alpha,\theta_i)\geq \bar{v}(\alpha,\hat{\theta}_i)$, then $\bar{\pi}_i(\theta_i)\geq \bar{\pi}_i(\hat{\theta}_i)$. Due to TIBS constraint in (4), we have
\begin{equation}
\theta_i\bar{v}_i(\alpha,\theta_i)-\bar{\pi}_i(\theta_i)\geq \theta_i\bar{v}_i(\alpha,\hat{\theta}_i)-\bar{\pi}_i(\hat{\theta}_i),
\end{equation}
which can be transformed to be 
\begin{equation}
\bar{\pi}_i(\hat{\theta}_i)-\bar{\pi}_i(\theta_i)\geq \theta_i\Big[\bar{v}_i(\alpha,\hat{\theta}_i)-\bar{v}_i(\theta_i)\Big].
\end{equation}
Since $\bar{v}_i(\alpha,\hat{\theta}_i)\geq\bar{v}_i(\theta_i)$, we conclude:
\begin{equation}
\bar{\pi}_i(\hat{\theta}_i)-\bar{\pi}_i(\theta_i)\geq \theta_i\Big[\bar{v}_i(\hat{\theta}_i)-\bar{v}_i(\alpha,\theta_i)\Big]\geq0,
\end{equation}
and thus $\bar{\pi}_i(\hat{\theta}_i)\geq\bar{\pi}_i(\theta_i)$.

\section{Proof of Theorem1}
First part of the proof ``if-part": First, we show that 1) and 2) in Theorem \ref{theorem} imply TIBS. Take any two values $\theta_i,\hat{\theta}_i\in[\ubar{\theta}_i,\bar{\theta}_i]$ with $\theta_i>\hat{\theta_i}>\ubar{\theta}$. From 2) we have:
\begin{equation}
\bar{u}_i(\theta_i)-\int_{\bar{\theta}_i}^{\theta_i}\bar{v}_i(\alpha,\eta)d\eta=\bar{u}_i(\hat{\theta}_i)-\int_{\hat{\theta}_i}^{\hat{\theta}_i}\bar{v}_i(\alpha,\eta)d\eta.
\end{equation}
Rearranging terms and using 1) gives:
\begin{flalign}
\bar{u}_i(\theta_i)-\bar{u}_i(\hat{\theta}_i)=\int_{\hat{\theta}_i}^{\theta_i}\bar{v}_i(\alpha,\eta)d\eta\geq\int_{\hat{\theta}_i}^{\theta_i}\bar{v}_i(\alpha,\hat{\theta}_i)d\eta=(\theta_i-\hat{\theta}_i)\bar{v}_i(\alpha,\hat{\theta}_i).
\end{flalign}
Note that $\bar{v}_i(\alpha,\hat{\theta}_i)$ is a constant. Hence,
\begin{equation}\label{1}
\bar{u}_i(\theta_i)\geq\bar{u}_i(\hat{\theta}_i)+(\theta_i-\hat{\theta}_i)\bar{v}_i(\alpha,\hat{\theta}_i)=\theta_i\bar{v}_i(\alpha,\hat{\theta}_i)+\bar{\pi}_i(\hat{\theta}_i).
\end{equation}
Similarly, suppose that $\hat{\theta}_i\geq \theta_i\geq\ubar{\theta}_i$. By the same reasoning, we have
\begin{equation}\label{2}
\bar{u}_i(\hat{\theta}_i)\geq \bar{u}_i(\theta_i)+(\hat{\theta}_i-\theta_i)\bar{v}_i(\alpha,\theta_i)=(\hat{\theta}_i-\theta_i)\bar{v}_i(\alpha,\theta_i).
\end{equation}

Together, (\ref{1}) and (\ref{2}) imply TIBS. Next, we prove that 1), 2) and 3) in Theorem \ref{theorem} imply IIR. We do this by contradiction. Suppose the implication is wrong. Then, there exists some $\theta_i\geq\ubar{\theta}_i$ with $\bar{u}_i(\theta_i)\leq0$. We just established that 1) and 2) imply TIBS. However, TIBS in conjunction with 3) implies:
\begin{equation}
\bar{u}_i(\theta_i)\geq \theta_i\bar{v}_i(\alpha,\ubar{\theta}_i)-\bar{\pi}_i(\ubar{\theta}_i)>\ubar{\theta}_i\bar{v}_i(\alpha,\ubar{\theta}_i)-\bar{\pi}_i(\ubar{\theta}_i),
\end{equation}
a contradiction.

"Only if part": We now show that TIBS implies 1), and 2). For any $i\in\mathcal{N}$ and any two types $\theta_i, \hat{\theta}_i\in[\ubar{\theta}_i,\bar{\theta}_i]$, TIBS requires that,
\begin{equation}\label{3}
\bar{u}_i(\theta_i)\geq \theta_i\bar{v}_i(\alpha,\hat{\theta}_i)-\bar{\pi}_i(\hat{\theta}_i)=\bar{u}_i(\hat{\theta}_i)+(\theta_i-\hat{\theta}_i)\bar{v}_i(\alpha,\hat{\theta}_i),
\end{equation}
and,
\begin{equation}\label{4}
\bar{u}_i(\hat{\theta}_i)\geq \hat{\theta}_i\bar{v}_i(\alpha,\theta_i)-\bar{\pi}_i(\theta_i)=\bar{u}_i(\theta_i)+(\theta_i-\hat{\theta}_i)\bar{v}_i(\alpha,\theta_i).
\end{equation}
Suppose without loss of generality that $\theta_{i}\leq\hat{\theta}_i$. From (\ref{3}) and (\ref{4}), it follows that
\begin{equation}
\bar{v}_i(\alpha,\hat{\theta}_i)\geq \frac{\bar{u}_i(\theta_i)-\bar{u}(\hat{\theta}_i)}{\theta_i-\hat{\theta}_i}\geq \bar{v}_i(\alpha,\theta_i),
\end{equation}
which shows that $\bar{v}_i(.)$ is nondecreasing. Next, letting $\hat{\theta}_i\to\theta_i$, we obtain $\frac{d\bar{u}_i(\theta_i)}{d\theta_i}=\bar{v}_i(\theta_i)$ for all $\theta_i$. Integrating both sides over $[\ubar{\theta}_i,\theta_i]$ gives
\begin{equation}
\bar{u}_i(\theta_i)=\bar{u}(\theta_i)+\int_{\ubar{\theta}_i}^{\theta_i}\bar{v}_i(\alpha,\eta)d\eta,
\end{equation}
$\forall \theta_i$. Finally, note that IIR obviously implies 3) by definition of IIR.

\section{Proof of Proposition \ref{prop}}
\subsection{Feasibility conditions:} We first show how we derive the price in (\ref{price}). We replace conditions (\ref{IR}) and (\ref{IC}) in (\ref{problem}) by the conditions provided in Theorem 1:
\begin{equation}
\label{pb}
\begin{aligned}
& \underset{\{\alpha_i(\boldsymbol{\theta}),\boldsymbol{\pi}_i(\boldsymbol{\theta})\}_{\forall i}}{\text{max}}
& &\int_{\Theta}\sum_{i\in\mathcal{N}}\left[\pi_i(\theta_i,\boldsymbol{\theta}_{-i})-c_i(\theta_i,\boldsymbol{\theta}_{-i})\right]dP(\theta),  \\
& \text{subject to}
& &  \bar{v}_i(\theta_i) \text{ is nondecreasing}, \\
&&& \bar{u}_i(\theta_i)=\bar{u}_i(\ubar{\theta}_i)+\int_{\ubar{\theta}_i}^{\theta_i}{\bar{v}_i(\alpha,\eta)d\eta}, \forall \theta_i, i\in\mathcal{N},\\
&&& \bar{u}_i(\ubar{\theta}_i)\geq 0, \forall i\in\mathcal{N},\\
&&& 0\leq\alpha_i(\boldsymbol{\theta})\leq 1,\forall i\in\mathcal{N}.
\end{aligned}
\end{equation}
After doing some algebraic manipulations on the second constraint in (\ref{pb}), we have:
\begin{flalign}
\label{cons}
-\int_{\Theta}\pi_i(\boldsymbol{\theta})=\ubar{\theta}_iv_i(\alpha,\ubar{\theta}_i)-\pi_i(\ubar{\theta}_i)-\int_{\Theta}\pi_i(\boldsymbol{\theta})\left(\theta_i-\frac{1-P_i(\theta_i)}{p_i(\theta_i)}\right)dP(\theta).
\end{flalign}
Inserting (\ref{cons}) in the objective function of the optimization problem (\ref{pb}), the problem can be relaxed and written as:
\begin{equation}
\label{relax}
\begin{aligned}
& \underset{\{\alpha_i(\boldsymbol{\theta}),\boldsymbol{\pi}_i(\boldsymbol{\theta})\}_{\forall i}}{\text{max}}
& &\int_{\Theta}\sum_{i\in\mathcal{N}}\left[\pi_i(\theta_i,\boldsymbol{\theta}_{-i})-c_i(\theta_i,\boldsymbol{\theta}_{-i})\right]dP(\theta),  \\
& \text{subject to}
& &  \bar{v}_i(\alpha,\theta_i) \text{ is nondecreasing}, \\
&&& \bar{u}_i(\ubar{\theta}_i)\geq 0, \forall i\in\mathcal{N},\\
&&& 0\leq\alpha_i(\boldsymbol{\theta})\leq 1,\forall i\in\mathcal{N}.
\end{aligned}
\end{equation}

We can see that the second constrain in (\ref{relax}) must bind, thus we have:
\begin{equation}
\label{e}
\bar{u}_i(\ubar{\theta}_i)=0,
\end{equation}
for all $i\in\mathcal{N}$. By Theorem \ref{theorem}, we can write (\ref{e}) as:
\begin{equation}
\bar{\pi}_i^*(\theta_i)= \theta_i\bar{v}_i(,\alpha_i^*,\theta_i)+\int_{\ubar{\theta}_i}^{\theta_i}\bar{v}_i(\alpha_i^*,\eta)d\eta,
\end{equation}
where we used $u_i(\theta_i)=\theta_i\bar{v}_i(\alpha_i,\theta_i)-\bar{\pi}_i(\theta_i)$. Next, we proof the optimality and the uniqueness of this solution.

\subsection{Proof of optimality:} We first show that the optimal price in (\ref{prop}) maximizes the utility of the MNO assuming that the amount of traffic served over the licensed and unlicensed bands is fixed to $\{\alpha_i^*, \forall i\}$. Suppose that that there exists another feasible price $\{\tilde{\pi}_i, \forall i\}$ which achieves a better solution than $\{\pi^*_i,\forall i\}$ in (\ref{prop}). Since the expected utility of the MNO given by the objective function in (\ref{problem}) is an increasing function of the total sum of the prices $\sum_{i\in\mathcal{N}}\pi_i$ payed by the users, we must have that $\sum_{i\in\mathcal{N}}\tilde{\pi}_i>\sum_{i\in\mathcal{N}}\pi^*_i$. Thus, for at least one user of type $\theta_i$, we have $\tilde{\pi}_i>\pi^*_i$ which is equivalent to $\overline{\tilde{\pi}}_i>\overline{\pi^*}_i$ because we have $\bar{\pi}_i(\theta_i,\boldsymbol{\theta}_{-i})=\int_{\Theta_{-i}}\pi_i(\theta_i,\boldsymbol{\theta}_{-i})dP_{-i}(\boldsymbol{\theta}_{-i})$. 

If $i=1$, then $\overline{\tilde{\pi}}_1>\overline{\pi^*}_1$. Since $\overline{\pi^*}_i=\theta_1\bar{v}_1(\alpha_i,\theta_i)$, then $\overline{\tilde{\pi}}_1>\theta_1\bar{v}_1(\alpha_i,\theta_i)$. However, it violates the IIR constraint for type $\theta_1$. Thus, we have to check if for any $i>1$ the supposition can be satisfied.

Since $\{\tilde{\pi}_i,\forall i\}$ is feasible, then $\{\tilde{\pi}_i,\forall i\}$ must satisfy the right inequality of  condition 3) in Theorem \ref{theorem}. Thus, we have:
\begin{equation}
\overline{\tilde{\pi}}_i\leq\overline{\tilde{\pi}}_{i-1}+\theta_i(\bar{v}_i(\alpha_i,\theta)-\bar{v}_{i-1}(\alpha_{i-1},\theta_{i-1})).
\end{equation}
By substituting $\theta_i(\bar{v}_i(\alpha_i,\theta)-\bar{v}_{i-1}(\alpha_{i-1},\theta_{i-1}))=\overline{\pi_i^*}-\overline{\pi^*}_{i-1}$  as in (\ref{prop}) into this inequality, we have
\begin{equation}
\overline{\tilde{\pi}}_{i-1}>\overline{\pi^*}_{i-1}.
\end{equation}
Using the above argument repeatedly, we finally obtain that
\begin{equation}
\overline{\tilde{\pi}}_1>\overline{\pi^*}_1=\theta_1\bar{v}_1(\alpha_1,\theta_1),
\end{equation}
which violates the IIR constraint for type-$\theta_1$ again.
\subsection{Proof of uniqueness:} Here, we prove that the price defined in (\ref{prop}) is the unique solution that maximizes the objective function in (\ref{problem}). We also prove by contradiction. Suppose that there exists another $\{\hat{\pi}_i,\forall i\}\neq\{\pi_i^*,\forall i\}$ such that $\sum_{i\in\mathcal{N}}{\overline{\hat{\pi}}_i}=\sum_{i\in\mathcal{N}}{\overline{\pi^*}_i}$ in the objective function of (\ref{problem}). 

Thus, there is at least one user for which the price satisfies, $\overline{\hat{\pi}}_i<\overline{\pi^*}_i$ and one user for which $\overline{\hat{\pi}}_l>\overline{\pi^*}_l$. We can focus on type $\theta_l$ and $\overline{\hat{\pi}}_l>\overline{\pi^*}_l$.
By using the same argument as before, we have $\overline{\hat{\pi}}_1>\overline{\pi^*}_1=\theta_1\bar{v}_1(\alpha_1,\theta_1)$ which violates the IIR constraint for type $\theta_1$.

\end{appendices}

\bibliographystyle{IEEEtran}
\bibliography{references}

\begin{thebibliography}{10}
\providecommand{\url}[1]{#1}
\csname url@samestyle\endcsname
\providecommand{\newblock}{\relax}
\providecommand{\bibinfo}[2]{#2}
\providecommand{\BIBentrySTDinterwordspacing}{\spaceskip=0pt\relax}
\providecommand{\BIBentryALTinterwordstretchfactor}{4}
\providecommand{\BIBentryALTinterwordspacing}{\spaceskip=\fontdimen2\font plus
\BIBentryALTinterwordstretchfactor\fontdimen3\font minus
  \fontdimen4\font\relax}
\providecommand{\BIBforeignlanguage}[2]{{%
\expandafter\ifx\csname l@#1\endcsname\relax
\typeout{** WARNING: IEEEtran.bst: No hyphenation pattern has been}%
\typeout{** loaded for the language `#1'. Using the pattern for}%
\typeout{** the default language instead.}%
\else
\language=\csname l@#1\endcsname
\fi
#2}}
\providecommand{\BIBdecl}{\relax}
\BIBdecl

\bibitem{cisco}
Cisco, ``Cisco visual networking index: Forecast and methodology, 2016-2021,''
  \emph{White Paper}.

\bibitem{zhanglte}
R.~Zhang, M.~Wang, L.~X. Cai, Z.~Zheng, X.~S. Shen, and L.-L. Xie,
  ``{LTE}-unlicensed: The future of spectrum aggregation for cellular
  networks,'' \emph{IEEE Wireless Communications}, vol.~22, no.~3, pp.
  150--159, June 2015.

\bibitem{zhang2015coexistence}
H.~Zhang, X.~Chu, W.~Guo, and S.~Wang, ``Coexistence of {Wi-Fi} and
  heterogeneous small cell networks sharing unlicensed spectrum,'' \emph{IEEE
  Communications Magazine}, vol.~53, no.~3, pp. 158--164, March 2015.

\bibitem{liu2014small}
F.~Liu, E.~Bala, E.~Erkip, M.~C. Beluri, and R.~Yang, ``Small cell traffic
  balancing over licensed and unlicensed bands,'' \emph{IEEE Transactions on
  Vehicular Technology}, vol.~64, no.~12, December 2015.

\bibitem{GuanINFOCOM16}
Z.~Guan and T.~Melodia, ``{{U-LTE}: Spectrally-Efficient and Fair Coexistence
  Between {LTE} and {Wi-Fi} in Unlicensed Bands},'' in \emph{Proc. of IEEE
  Conference on Computer Communications (INFOCOM)}, San Francisco, CA, April
  2016.

\bibitem{khairy2017hybrid}
S.~Khairy, L.~X. Cai, Y.~Cheng, Z.~Han, and H.~Shan, ``A {Hybrid-LBT MAC} with
  adaptive sleep for {LTE LAA} coexisting with {Wi-Fi} over unlicensed band,''
  in \emph{Proc. IEEE Global Communications Conference (GLOBECOM)}, Singapore,
  December 2017.

\bibitem{bairagi2018multi}
A.~K. Bairagi, N.~H. Tran, and C.~S. Hong, ``A multi-game approach for
  effective co-existence in unlicensed spectrum between {LTE-U} system and
  {Wi-Fi} access point,'' in \emph{Proc. International Conference on
  Information Networking (ICOIN)}, Chiang Mai, Thailand, April 2018, pp.
  380--385.

\bibitem{bennis2013cellular}
M.~Bennis, M.~Simsek, A.~Czylwik, W.~Saad, S.~Valentin, and M.~Debbah, ``When
  cellular meets {WiFi} in wireless small cell networks,'' \emph{IEEE
  Communications Magazine,}, vol.~51, no.~6, pp. 44--50, June 2013.

\bibitem{chen2016echo}
M.~Chen, W.~Saad, and C.~Yin, ``Echo state networks for self-organizing
  resource allocation in {LTE-U} with uplink-downlink decoupling,'' \emph{IEEE
  Transactions on Wireless Communications}, vol.~16, no.~1, pp. 3--16, January
  2017.

\bibitem{gu2015exploiting}
Y.~Gu, Y.~Zhang, L.~X. Cai, M.~Pan, L.~Song, and Z.~Han, ``Exploiting
  student-project allocation matching for spectrum sharing in
  {LTE}-unlicensed,'' in \emph{Proc. IEEE Global Communications Conference},
  San Diego, CA, February 2015.

\bibitem{elsherif2015resource}
A.~R. Elsherif, W.-P. Chen, A.~Ito, and Z.~Ding, ``Resource allocation and
  inter-cell interference management for dual-access small cells,'' \emph{IEEE
  Journal on Selected Areas in Communications}, vol.~33, no.~6, pp. 1082--1096,
  June 2015.

\bibitem{kang2014mobile}
X.~Kang, Y.-K. Chia, S.~Sun, and H.~F. Chong, ``Mobile data offloading through
  a third-party {WiFi} access point: An operator's perspective,'' \emph{IEEE
  Transactions on Wireless Communications}, vol.~13, no.~10, pp. 5340--5351,
  October 2014.

\bibitem{chen2016impact}
C.~Chen, R.~A. Berry, M.~L. Honig, and V.~G. Subramanian, ``The impact of
  unlicensed access on small-cell resource allocation,'' \emph{IEEE
  International Conference on Computer Communications}, April 2016.

\bibitem{challita2017proactive}
U.~Challita, L.~Dong, and W.~Saad, ``Proactive resource management in {LTE-U}
  systems: A deep learning perspective,'' \emph{IEEE Transactions on Wireless
  Communications}, vol.~17, no.~7, pp. 4674 -- 4689, July 2018.

\bibitem{xiao2018optimizing}
Y.~Xiao, M.~Hirzallah, and M.~Krunz, ``Optimizing inter-operator network
  slicing over licensed and unlicensed bands,'' in \emph{Proc. 15th Annual IEEE
  International Conference on Sensing, Communication, and Networking (SECON)},
  Hong Kong, China, June 2018.

\bibitem{garcia2018orla}
A.~Garcia-Saavedra, P.~Patras, V.~Valls, X.~Costa-Perez, and D.~J. Leith,
  ``Orla/olaa: Orthogonal coexistence of laa and wifi in unlicensed spectrum,''
  \emph{arXiv preprint arXiv:1802.01360}, 2018.

\bibitem{gu2017dynamic}
Y.~Gu, C.~Jiang, L.~X. Cai, M.~Pan, L.~Song, Z.~Han \emph{et~al.}, ``Dynamic
  path to stability in {LTE}-unlicensed with user mobility: A matching
  framework,'' \emph{IEEE Transactions on Wireless Communications}, vol.~16,
  no.~7, pp. 4547--4561, May 2017.

\bibitem{mozaffari2016unmanned}
M.~Mozaffari, W.~Saad, M.~Bennis, and M.~Debbah, ``Unmanned aerial vehicle with
  underlaid device-to-device communications: Performance and tradeoffs,''
  \emph{IEEE Transactions on Wireless Communications}, vol.~15, no.~6, pp.
  3949--3963, June 2016.

\bibitem{challita2017network}
U.~Challita and W.~Saad, ``Network formation in the sky: Unmanned aerial
  vehicles for multi-hop wireless backhauling,'' in \emph{IEEE Global
  Communications Conference, Next Generation Networking and the Internet
  Symposium}, Singapore, December 2017.

\bibitem{elec2015fair}
N.~Garg, ``Fair use and innovation in unlicensed wireless spectrum,''
  \emph{Stanford University}, 2015.

\bibitem{shih2018unlicensed}
M.-J. Shih, T.-H. Wu, and H.-Y. Wei, ``Unlicensed {LTE} pricing for tiered
  content delivery and heterogeneous user access,'' \emph{IEEE Transactions on
  Mobile Computing, to appear}.

\bibitem{yu2017auction}
H.~Yu, G.~Iosifidis, J.~Huang, and L.~Tassiulas, ``Auction-based coopetition
  between {LTE} unlicensed and {Wi-Fi},'' \emph{IEEE Journal on Selected Areas
  in Communications}, vol.~35, no.~1, pp. 79--90, January 2017.

\bibitem{zhang2017multi}
H.~Zhang, Y.~Xiao, L.~X. Cai, D.~Niyato, L.~Song, and Z.~Han, ``A multi-leader
  multi-follower stackelberg game for resource management in {LTE}
  unlicensed,'' \emph{IEEE Transactions on Wireless Communications}, vol.~16,
  no.~1, pp. 348--361, January 2017.

\bibitem{borgers2015introduction}
T.~Borgers, R.~Strausz, and D.~Krahmer, \emph{An introduction to the theory of
  mechanism design}.\hskip 1em plus 0.5em minus 0.4em\relax Oxford University
  Press, USA, 2015.

\bibitem{zhang2017survey}
Y.~Zhang, M.~Pan, L.~Song, Z.~Dawy, and Z.~Han, ``A survey of contract
  theory-based incentive mechanism design in wireless networks,'' \emph{IEEE
  Wireless Communications}, vol.~24, no.~3, pp. 80--85, June 2017.

\bibitem{gao2011spectrum}
L.~Gao, X.~Wang, Y.~Xu, and Q.~Zhang, ``Spectrum trading in cognitive radio
  networks: A contract-theoretic modeling approach,'' \emph{IEEE Journal on
  Seleceted Areas in Communications}, vol.~29, no.~4, pp. 843--855, April 2011.

\bibitem{duan2014cooperative}
L.~Duan, L.~Gao, and J.~Huang, ``Cooperative spectrum sharing: a contract-based
  approach,'' \emph{IEEE Transactions on Mobile Computing}, vol.~13, no.~1, pp.
  174--187, January 2014.

\bibitem{zhang2015contract}
Y.~Zhang, L.~Song, W.~Saad, Z.~Dawy, and Z.~Han, ``Contract-based incentive
  mechanisms for device-to-device communications in cellular networks,''
  \emph{IEEE Journal on Selected Areas in Communications}, vol.~33, no.~10, pp.
  2144--2155, October 2015.

\bibitem{hamidouche2016breaking}
K.~Hamidouche, W.~Saad, and M.~Debbah, ``Breaking the economic barrier of
  caching in cellular networks: Incentives and contracts,'' in \emph{Proc. IEEE
  Global Communication Conference}, Washington, DC, December 2016.

\bibitem{shah2017incentive}
H.~Shah-Mansouri, V.~W. Wong, and J.~Huang, ``An incentive framework for mobile
  data offloading market under price competition,'' \emph{IEEE Transactions on
  Mobile Computing}, vol.~16, no.~11, pp. 2983--2999, November 2017.

\bibitem{hajimirsadeghi2017joint}
M.~Hajimirsadeghi, N.~B. Mandayam, and A.~Reznik, ``Joint caching and pricing
  strategies for popular content in information centric networks,'' \emph{IEEE
  Journal on Selected Areas in Communications}, vol.~35, no.~3, pp. 654--667,
  March 2017.

\bibitem{liu2017design}
T.~Liu, J.~Li, F.~Shu, M.~Tao, W.~Chen, and Z.~Han, ``Design of contract-based
  trading mechanism for a small-cell caching system,'' \emph{IEEE Transactions
  on Wireless Communications}, vol.~16, no.~10, pp. 6602--6617, October 2017.

\bibitem{liu2014stable}
Q.~Liu, G.~J. Mailath, A.~Postlewaite, and L.~Samuelson, ``Stable matching with
  incomplete information,'' \emph{Econometrica}, vol.~82, no.~2, pp. 541--587,
  April 2014.

\bibitem{wang2013interdependent}
X.~Y. Wang, ``Interdependent utility and truthtelling in two-sided matching,''
  \emph{Economic Research Initiatives at Duke (ERID) Working Paper}, vol. 159,
  2013.

\bibitem{bodine2011peer}
E.~Bodine-Baron, C.~Lee, A.~Chong, B.~Hassibi, and A.~Wierman, ``Peer effects
  and stability in matching markets,'' in \emph{International Symposium on
  Algorithmic Game Theory}.\hskip 1em plus 0.5em minus 0.4em\relax Springer,
  2011, pp. 117--129.

\bibitem{gu2015matching}
Y.~Gu, W.~Saad, M.~Bennis, M.~Debbah, and Z.~Han, ``Matching theory for future
  wireless networks: fundamentals and applications,'' \emph{IEEE Communications
  Magazine}, vol.~53, no.~5, pp. 52--59, May 2015.

\bibitem{Roth1990}
A.~E. Roth and M.~Sotomayor, \emph{Two-Sided Matching: A Study in
  Game-Theoretic Modeling and Analysis}, ser. Econometric Society
  Monograph.\hskip 1em plus 0.5em minus 0.4em\relax Cambridge University Press,
  1990.

\bibitem{Roth1991}
A.~E. Roth, ``A natural experiment in the organization of entry level labor
  markets: Regional markets for new physicians and surgeons in the u.k,''
  \emph{American Economic Review}, vol.~81, pp. 415--425, June 1991.

\bibitem{semiari2014matching}
O.~Semiari, W.~Saad, S.~Valentin, M.~Bennis, and B.~Maham, ``Matching theory
  for priority-based cell association in the downlink of wireless small cell
  networks,'' in \emph{Proc. IEEE International Conference on Acoustics, Speech
  and Signal Processing (ICASSP)}, Florence, Italy, May 2014, pp. 444--448.

\bibitem{xiao2015bayesiano}
Y.~Xiao, K.-C. Chen, C.~Yuen, Z.~Han, and L.~A. DaSilva, ``A bayesian
  overlapping coalition formation game for device-to-device spectrum sharing in
  cellular networks,'' \emph{IEEE Transactions on Wireless Communications},
  vol.~14, no.~7, pp. 4034--4051, July 2015.

\bibitem{xiao2015bayesian}
Y.~Xiao, Z.~Han, K.-C. Chen, and L.~A. DaSilva, ``Bayesian hierarchical
  mechanism design for cognitive radio networks,'' \emph{IEEE Journal on
  Selected Areas in Communications}, vol.~33, no.~5, pp. 986--1001, May 2015.

\bibitem{roth1992two}
A.~E. Roth and M.~A.~O. Sotomayor, \emph{Two-sided matching: A study in
  game-theoretic modeling and analysis}.\hskip 1em plus 0.5em minus 0.4em\relax
  Cambridge University Press, 1992, no.~18.

\bibitem{pantisano2013matching}
F.~Pantisano, M.~Bennis, W.~Saad, S.~Valentin, and M.~Debbah, ``Matching with
  externalities for context-aware user-cell association in small cell
  networks,'' in \emph{Proc. IEEE Globecom Workshops (GC Wkshps)}, Atlanta, GA,
  December 2013.

\end{thebibliography}

\end{document}